\documentclass[journal,draftclsnofoot,onecolumn,12pt,twoside]{IEEEtran}

\usepackage{graphicx}
\usepackage{amsmath}
\usepackage{amsthm}
\usepackage{amsfonts}
\usepackage[justification=centering]{caption}
\usepackage{color, xcolor}
\usepackage{diagbox}
\usepackage{bm}
\usepackage{booktabs}
\usepackage{multirow}
\usepackage{algorithm}
\usepackage{algpseudocode}
\usepackage{amssymb}
\usepackage{cite}
\usepackage{hyperref} 
\usepackage{caption}
\usepackage{subcaption}
\usepackage{pstricks}
\usepackage{rotating}
\usepackage{ulem}
\usepackage{blkarray} 
\allowdisplaybreaks[4]

\theoremstyle{}
\newtheorem{theorem}{Theorem}
\newtheorem{lemma}{Lemma}
\newtheorem{definition}{Definition}
\newtheorem{corollary}{Corollary}

\newtheorem{example}{Example}
\newtheorem{remark}{Remark}
\newtheorem{construction}{Construction}

\newcommand{\tabcaption}{\def\@captype{table}\caption}

\makeatletter
\newcommand*\bigcdot{\mathpalette\bigcdot@{.5}}
\newcommand*\bigcdot@[2]{\mathbin{\vcenter{\hbox{\scalebox{#2}{$\m@th#1{\bullet}$}}}}}

\makeatother

\normalsize
\ifCLASSINFOpdf

\else

\fi

\begin{document}
\title{Coded Caching for Two-Dimensional Multi-Access Networks
\author{Mingming Zhang,  Kai~Wan,~\IEEEmembership{Member,~IEEE,} Minquan Cheng
and~Giuseppe~Caire,~\IEEEmembership{Fellow,~IEEE}}
\thanks{M. Zhang and M. Cheng are with Guangxi Key Lab of Multi-source Information Mining $\&$ Security, Guangxi Normal University,
Guilin 541004, China  (e-mail: ztw\_07@foxmail.com, chengqinshi@hotmail.com). }
\thanks{K. Wan and G. Caire are with the Electrical Engineering and Computer Science Department, Technische Universit\"{a}t Berlin,
10587 Berlin, Germany (e-mail: kai.wan@tu-berlin.de, caire@tu-berlin.de).  The work of K.~Wan and G.~Caire was partially funded by the European Research Council under the ERC Advanced Grant N. 789190, CARENET.}
}
\maketitle

\begin{abstract}
This paper studies a novel multi-access coded caching (MACC) model in the two-dimensional (2D) topology, which is a generalization of the one-dimensional (1D) MACC model proposed by Hachem
{\it et al.} The 2D MACC model is formed by a server containing $N$ files, $K_1\times K_2$ cache-nodes with $M$ files located at a grid with $K_1$ rows and $K_2$ columns, and $K_1\times K_2$ cache-less users where each user is connected to $L^2$ nearby cache-nodes. The server is connected to the users through an error-free shared link, while the users can retrieve the cached content of the connected cache-nodes without cost.  Our objective is to minimize the  worst-case transmission load over all possible users' demands. In this paper, we first propose a \textit{grouping scheme} for the case where   $K_1$ and $K_2$ are divisible by $L$. By partitioning
  the cache-nodes and users into $L^2$ groups   such that no two users in the same group share any cache-node, we use the shared-link coded caching scheme proposed by Maddah-Ali and Niesen for each group. Then for any model parameters satisfying $\min\{K_1,K_2\}>L$, we propose a transformation approach which constructs a 2D MACC scheme from two classes of 1D MACC schemes in vertical and  horizontal projections, respectively. 
  As a result, we can construct 2D MACC schemes that achieve maximum local caching gain and improved coded caching gain, compared to the baseline scheme by a direct extension from 1D MACC schemes.
\end{abstract}

\begin{IEEEkeywords}
	Coded caching, multi-access coded caching, two-dimensional (2D) network, placement delivery array (PDA).
\end{IEEEkeywords}

\section{Introduction}
Caching techniques have a central role in future communication systems and wireless cellular  networks\cite{GGTTG}. In the caching paradigm,  some content is locally stored into the users' local caches during off-peak times. Then the pre-stored content is leveraged to reduce the network congestion during peak times, such that some local caching gain arises. In the seminal paper \cite{MN},  Maddah-Ali and Niesen (MN) proposed a coded caching scheme which achieves an additional multicast gain on top of the conventional local caching gain.  In the MN coded caching model,  a single server with  $N$ file is connected to $K$ users over an error-free shared link, while  each user has a local cache of size $M$.
 A coded caching scheme consists of two phases: i) \textit{placement phase}: some packets of each file are placed into the cache of each user without knowledge of the user's future demand; ii) \textit{delivery phase}: each user requests one file. According to the users' demands and    cache content, the server sends coded packets such that each user's demand is satisfied. The goal is to minimize the worst-case number of transmitted packets normalized by the file size (referred to as \textit{load} in this paper).

The MN coded caching scheme utilizes an uncoded combinatorial cache construction in the placement phase and linear coding in the delivery phase. When $M=t\frac{N}{K}$ with $t\in \{0,1,\ldots,K\}$, the achieved load is $\frac{K(1-M/N)}{1+KM/N}$.
The term $1-M/N$ in the numerator is the \textit{local caching gain}, which
is defined as the average fraction of each file not available in the cache of each user.
The term $1+KM/N$ in the denominator is the \textit{coded caching gain}, which is defined as the average number of users served by one multicast message.
For other memory sizes, the lower convex envelope of the above memory-load tradeoff can be achieved by memory-sharing.
The load of the MN coded caching scheme was proved to be optimal within a factor of $2$~\cite{yufactor2TIT2018} and exactly optimal under the constraint of uncoded cache placement~\cite{WTP,YMA}  (i.e., each user directly copies some packets of files in its cache).


However, the MN scheme requires a subpacketization exponential to the number of users $K$. In order to reduce the subpacketization, the authors in \cite{YCTC} proposed a combinatorial structure to characterize the placement and delivery strategies in a single array, referred to as \emph{Placement Delivery Array (PDA)}. It was shown in \cite{STD}  that the schemes in \cite{MN,YCTC,STD,SDLT,SZG,TR,YTCC,CKSB} can be represented by appropriate PDAs.
Particularly, the PDA characterizing the MN scheme in \cite{MN} is referred to as MN PDA.
 Given any $(K,F,Z,S)$ PDA, we can obtain a shared-link coded caching scheme for $K$ users, with subpacketization $F$, memory ratio $\frac{M}{N}=\frac{Z}{F}$ and load $R=\frac{S}{F}$. By using PDA, various coded caching schemes were constructed to reduce the subpacketization of the MN scheme, e.g.,~\cite{YCTC,CJWY,CJTY,CJYT,MW,ZCJ,YTCC,SZG,CWLZC,SB,ER,ENR}.

\subsection{One-Dimension Multi-access Caching}
Most   works on coded caching consider that each user has its own dedicated cache. Edge caching, which stores the  Internet-based content at  the wireless edges,  boosts the spatial and spectral efficiency. The main advantages of edge caching compared to the end-user caches include  that the edge nodes normally have larger storage sizes and could be accessed by multiple local users with high data rates.
Such a scenario motivated the work in \cite{JHNS} which introduced a multi-access coded caching (MACC) problem, referred to as $(K,L,M,N)$ one-dimensional (1D) MACC problem. Different from the MN coded caching problem,    there are $K$ cache-nodes with the cache size of $M$ files, while  each of the $K$ users is cache-less and  can access $L$ neighboring cache-nodes in a cyclic wrap-around fashion.  Thus,  each cache-node serves exactly $L$ users. As assumed in \cite{JHNS}, the cache-nodes  are accessed by the connected users with negligible load cost.

Under the 1D MACC model, various schemes were proposed in \cite{JHNS,SPE,RK,SR,RKstructure,LWCG,NRprivacy,NRsecure,OG,KMR,MKR,MKRmn,FP}.
The most related work to this paper is our previous work in \cite{CWLZC}, which  proposed a transformation approach to extend any  PDA  for the shared-link coded caching system (satisfying some constraints which most existing PDAs satisfy) to generate a 1D MACC scheme as illustrated in Fig.~\ref{fig-intr-1D}.
\begin{figure}
	\centering
	\includegraphics[width=4.3in]{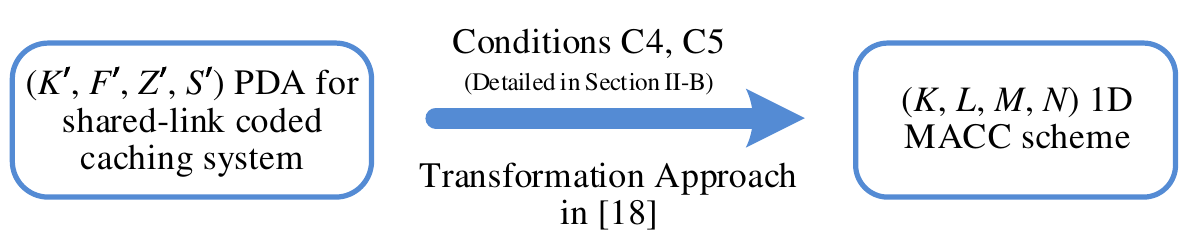}
	\caption{Transformation from a $(K',F',Z',S')$ shared-link PDA to a $(K,L,M,N)$ 1D MACC scheme, where $K=K'+\frac{K'Z'}{F'}(L-1)$ and $M=\frac{K'Z'N}{F'K}$.}\label{fig-intr-1D}
\end{figure}
For any $(K,L,M,N)$ 1D MACC system with $M/N\in \{0,1,\ldots, \left\lfloor K/L\right\rfloor \}$, by using such transformation approach on the MN PDA for shared-link coded caching scheme, we can obtain a  1D MACC caching scheme with the load $\frac{K(1-LM/N)}{KM/N+1}$.

\subsection{Two-Dimensional Multi-access Caching}
The aforementioned works on the MACC problem only considered the 1D topology. However, in a practical cellular network, the cache-nodes are most typically placed in a two-dimensional (2D) topology to cover a plane area, such as triangle, square, and hexagon cellular geometries\cite{MV1979}. Motivated by this, we consider an ideal MACC problem with 2D square topology, referred to as 2D MACC. In this paper, we focus on the $(K_1,K_2,L,M,N)$ 2D MACC system as illustrated in Fig.~\ref{fig-grid-intru}. In this setting, $K_1\times K_2$ cache-nodes with cache size of $M$ files are placed in a rectangular grid with $K_1$ rows and $K_2$ columns, and $K_1\times K_2$ cache-less users are placed regularly on the same grid such that each user is in the proximity of a square of $L\times L$ neighboring cache-nodes (where distance is defined in a cyclic wrap-around fashion).    
\begin{figure}
	\centering
	\includegraphics[width=3.3in]{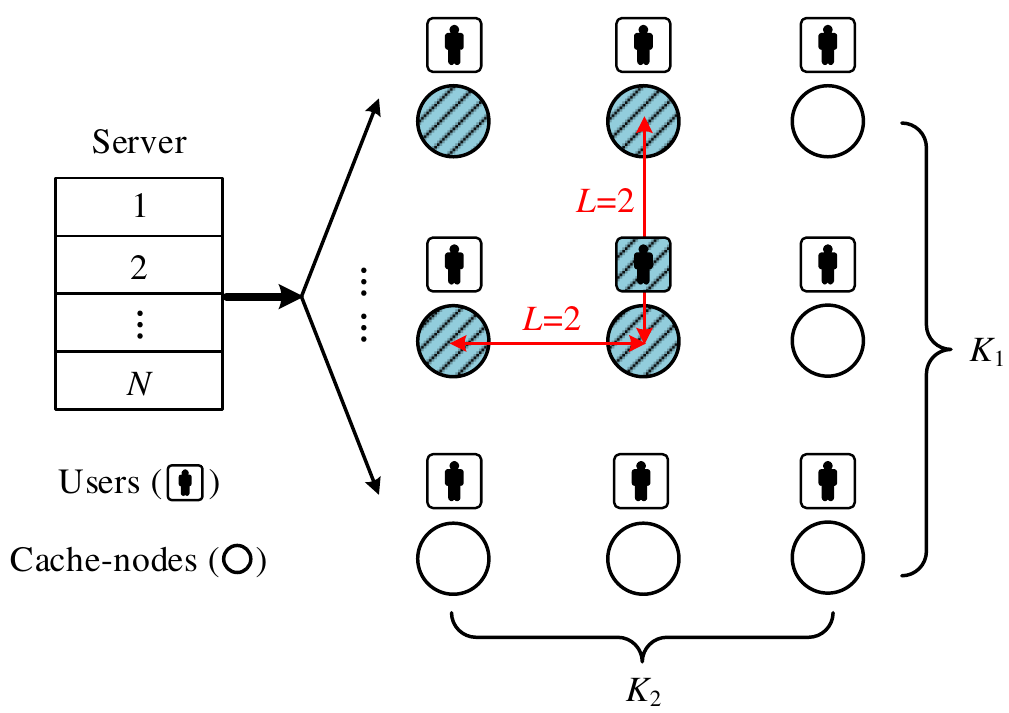}
	\caption{The 2D MACC model with $K_1=K_2=3$, $L=2$.}\label{fig-grid-intru}
\end{figure}
For instance, when $K_1=K_2=3$ and $L=2$, the user at row $2$ and column $2$ can access the cache-nodes which are located at (row, column)$=(1,1)$, $(1,2)$, $(2,1)$ and $(2,2)$ (as illustrated in Fig.~\ref{fig-grid-intru}). Without loss of generality, we assume that $K_1\geq K_2$. When  $K_2=1$, the 2D MACC system   reduces to 1D MACC system. Similar to the 1D MACC model, users can access their $L^2$ cache-nodes at no cost (this assumes very fast off-load side links between users and cache-nodes). The objective of the problem is to minimize the worst-case load of the broadcast transmissions from the server to users over all possible demands. 

\subsection{Contribution and Paper Organization}
Our contributions for  the new $(K_1,K_2,L,M,N)$ 2D MACC system are as follows. 
\begin{itemize}
\item 
 We first propose a \textit{baseline scheme}, by  using an MDS precoding on each file such that the 2D MACC problem is divided into $K_2$ separate 1D MACC problems, each of which has $K_1$ cache-nodes and users.  The baseline scheme achieves the maximum local caching gain 
and a coded caching gain equals to $\frac{K_1\min\{K_2,L\} M}{N}+1$.
\item   When $K_1$ and $K_2$ are divisible by $L$,
we propose a grouping scheme which partitions all the cache-nodes and  users into $L^2$ groups such that any two users in the same group cannot access the same cache-node, and uses the MN caching scheme for each group. The grouping scheme achieves the maximum local caching gain and a coded caching gain which equals to $\frac{K_1K_2 M}{N}+1$.
\item   Our major contribution on this new model is to propose a new transformation approach for the case   $K_2>L$, which constructs a hybrid  2D MACC scheme (i.e., consisting of an outer structure and an inner structure) from two classes of 1D MACC schemes as illustrated in Fig.~\ref{fig-contr-hybrid}. In the vertical projection of the 2D system which reduces to the $(K_1,L,M_1,N)$ 1D MACC system,
we select a  1D MACC scheme as outer structure from any 1D MACC scheme,  which is generated by the transformation approach\cite{CWLZC}.
In the horizontal projection of the 2D system which reduces to the $(K_2,L,M_2,N)$ 1D MACC system, we use  $\frac{K_1K_2M}{N}$ 1D MACC schemes as inner code, which are generated by  using the transformation approach\cite{CWLZC}  on the Partition PDA in \cite{CJYT} for the shared-link caching model.
Finally, by incorporating the outer and inner structures, we obtain a hybrid $(K_1,K_2,L,M,N)$ 2D MACC scheme, where  $\frac{M_1}{N}\cdot\frac{M_2}{N}=\frac{M}{N}$. The grouping scheme achieves the maximum local caching gain and a coded caching gain no less than $\frac{K_1K_2 M}{N}$ while the outer structure is generated based on the MN PDA.
\end{itemize}
\begin{figure}
	\centering
	\includegraphics[width=3.9in]{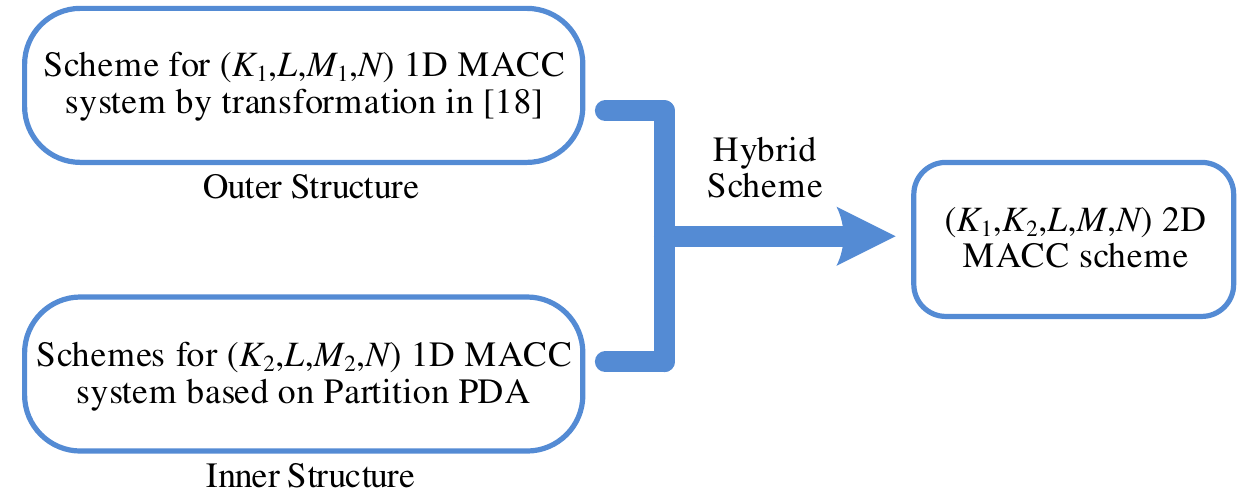}
	\caption{Contribution of hybrid scheme in 2D MACC model}\label{fig-contr-hybrid}
\end{figure}

The rest of this paper is organized as follows. Section~\ref{sec:preliminary} reviews some related results on the original shared-link coded caching model and 1D MACC model. Section \ref{sect-system} formulates the novel 2D MACC model. Section \ref{sec-main result} lists the main results of the paper.    Sections~\ref{sec-proof of Theorem Group} and~\ref{sec-proof of Theorem 1} provide the detailed constructions of the proposed caching schemes. Finally, we conclude the paper in Section \ref{sec-conclusion} and some proofs are provided in the Appendices.

\subsection*{Notations}
In this paper,  we use the following notations unless otherwise stated.
\begin{itemize}
\item Bold capital letter, bold lower case letter and curlicue font will be used to denote array, vector and set respectively.
The $r^{\text{th}}$ element of a set represents the  $r^{\text{th}}$ smallest element in this set.  $|\cdot|$ is used to represent the cardinality of a set or
the length of a vector;
\item For any positive integers $a$, $b$, $t$ with $a<b$ and $t\leq b $,   non-negative set $\mathcal{V}$, and vector ${\bf e}$,
\begin{itemize}
\item   $[a]:=\{1,2,\ldots,a\}$, $[a:b] :=\{a,a+1,\ldots,b\}$, $[a:b):=\{a,a+1,\ldots,b-1\}$ and ${[b]\choose t}:=\{\mathcal{V}\ |\   \mathcal{V}\subseteq [b], |\mathcal{V}|=t\}$, i.e., ${[b]\choose t}$ is the collection of all $t$-sized subsets of $[b]$;
\item $mod(a,q)$ denotes the least non-negative residue of $a$ modulo $q$.
\item
$
<a>_q:=\begin{cases}
       mod(a,q)\ \ \ \text{if}\ \ mod(a,q)\neq0\\
       \ \ \ q \ \ \ \ \ \ \ \ \ \ \text{if}\ \ mod(a,q)=0\\
       \end{cases}
$
\item $[a:b]_q:=\{<a>_q,<a+1>_q,\ldots,<b>_q\}$.
\item $\mathcal{V}[h]$ represents the $h^{\text{th}}$ smallest element of $\mathcal{V}$, where $h\in[|\mathcal{V}|]$. Assuming that $\mathcal{V}[h]=k$, we use $\mu(k)$ to represent the order of $k$ in $\mathcal{V}$, i.e., $\mu(k)=h$ if and only if $\mathcal{V}[h]=k$ for any $k\in \mathcal{V}$. ${\bf e}|_{h}$ is the $h^{\text{th}}$ entry of ${\bf e}$ for each $h\in [|{\bf e}|]$;
\item $\mathcal{V}+a:=\{v+a\ |\ \forall \ v\in  \mathcal{V}\}$.
\item For any array  $\mathbf{P}$ with dimension $m\times n$, $\mathbf{P}(i,j)$ represents the element located at the $i^{\text{th}}$ row and the $j^{\text{th}}$ column of $\mathbf{P}$.
\item The matrix $[a;b]$ is written in a Matlab form, representing
$\begin{bmatrix}
	a \\
	b
\end{bmatrix}$.
\end{itemize}
\end{itemize}

\section{Preliminary Results on Original Coded Caching Model and 1D Multi-Access Coded Caching Model}
\label{sec:preliminary}
In this section, we review the original shared-link coded caching model in \cite{MN}  and the   PDA structure in \cite{YCTC,CJYT}. Then we review the 1D MACC model in \cite{JHNS} and the transformation approach in~\cite{CWLZC} which constructs 1D MACC schemes from PDAs.

\subsection{Original Shared-link Coded Caching Model}
\label{sub:ori model}
In the original coded caching model~\cite{MN}, referred to as shared-link coded caching model, a server containing $N$ equal-length files, $\mathcal{W}=\{W_{1}, W_{2}, \ldots, $ $W_{N}\}$, connects through an error-free shared link to $K$ users $U_1$, $U_2$, $\ldots$, $U_{K}$ with $K\leq N$. Each user has a cache with size of $M$ files where $0\leq M \leq N$.
An $F$-division $(K,M,N)$ coded caching scheme contains two phases.
\begin{itemize}
 \item {\bf Placement phase:} The server divides each file into $F$ packets with equal size, i.e., $W_{n}=\{W_{n,j}\ |\ j\in [F]\}$, then directly places up to $MF$ packets to each user's cache. Note that in this phase the server has no information of the  users' later demands. Define $\mathcal{Z}_k$ as the cache content of user $k$.
 \item {\bf Delivery phase:} Each user randomly requests one file from the server. Assume that the demand vector is $\mathbf{d}=(d_1,d_2,\cdots,d_{K})$, i.e.,    user $U_k$ requests $W_{d_k}$, where $d_k\in[N]$ and $k\in [K]$. According to the users' cache content and demand vector, the server broadcasts $S_{{\bf d}}$ coded packets to the users such that each user can decode its desired file.
\end{itemize}
The objective is to minimize the worst-case load   among all possible requests, defined as
\begin{align}
R=\max\left\{\frac{ S_{\mathbf{d}}}{F}\ \Big|\ \mathbf{d}\in[N]^K\right\}.  \label{eq:def of load}
\end{align} 

The authors in \cite{YCTC} proposed a combinatorial coded caching structure, referred to as  placement delivery array (PDA).
\begin{definition}\rm(\cite{YCTC})
\label{def-PDA}
For any positive integers $K$, $F$, $Z$ and $S$, an $F\times K$ array $\mathbf{P}$ composed of a specific symbol $``*"$ and $S$ integers in $[S]$, is called a $(K,F,Z,S)$ PDA if it satisfies the following conditions,
 \item [C$1$.] The symbol $``*"$ appears $Z$ times in each column;
 \item [C$2$.] Each integer in $[S]$ occurs at least once in the array;
 \item [C$3$.] For any two distinct entries $\mathbf{P}(j_1,k_1)$ and $\mathbf{P}(j_2,k_2)$, if $\mathbf{P}(j_1,k_1)=\mathbf{P}(j_2,k_2)=s\in[S]$, then $\mathbf{P}(j_1,k_2)=\mathbf{P}(j_2,k_1)=*$, i.e., the corresponding $2\times 2$  sub-array formed by rows $j_1,j_2$ and columns $k_1,k_2$ must be one of the following form
 \begin{align*}
 \left(\begin{array}{cc}
 s & *\\
 * & s
 \end{array}\right)~\textrm{or}~
 \left(\begin{array}{cc}
 * & s\\
 s & *
 \end{array}\right).
 \end{align*}
\hfill $\square$
\end{definition}
Notice that, for the sake of ease notation, sometimes we also express the non-star entries in a PDA by sets or vectors rather than integers.

Based on a $(K,F,Z,S)$ PDA, an $F$-division coded caching scheme for the $(K,M,N)$ coded caching system can be obtained in the following way.
\begin{itemize}
\item The columns represent the user indices while the rows represent the packet indices.
\item If $\mathbf{P}(j,k)=*$,   user $k$ caches the $j^{\text{th}}$ packet of all files. So, Condition C1 of Definition \ref{def-PDA} implies that all  users have the same   memory ratio  $\frac{M}{N}=\frac{Z}{F}$.
\item If $\mathbf{P}(j,k)$ is an integer $s$, the $j^{\text{th}}$ packet of each file is not stored by user $k$. Then the server transmits a multicast message (i.e., the XOR of all the requested packets indicated by $s$) to the users at time slot $s$. Condition C3 of Definition \ref{def-PDA} guarantees that each user can recover its requested packets since it has cached all the other packets in the multicast message except its requested one. The occurrence number of integer $s$ in $\mathbf{P}$, denoted by $g_s$, is the coded caching gain at time slot $s$, meaning that  the coded packet is broadcasted at the time slot $s$ and simultaneously useful for $g_s$ users.  $\mathbf{P}$ is said to be a $g$-$(K,F,Z,S)$ PDA if $g_s=g$ for all $s\in [S]$.
\item Condition C2 of Definition \ref{def-PDA} implies that the number of multicast messages transmitted by the server is  $S$; thus the load is $R=\frac{S}{F}$.
\end{itemize}

\begin{example}\label{MN-pda}
\rm
We use the following $g$-$(K,F,Z,S)=3$-$(3,3,2,1)$ PDA $\mathbf{P}$ to construct a $(K,M,N)=(3,2,3)$ coded caching scheme for the shared-link coded caching model.
\begin{eqnarray*}
\mathbf{P}=\left(
             \begin{array}{ccc}
               * & * & 1 \\
               * & 1 & * \\
               1 & * & * \\
             \end{array}
           \right)
\end{eqnarray*}
\begin{itemize}
 \item \textbf{Placement Phase}: The server divides each file into $3$ equal-size packets, i.e., $W_n=\{W_{n,1},W_{n,2},W_{n,3}\}$, $n\in [3]$. The users cache the following packets,
 \begin{align*}
 \mathcal{Z}_1=\left\{W_{n,1},W_{n,2}\ |\ n\in[3]\right\},\ \ \ \mathcal{Z}_2=\left\{W_{n,1},W_{n,3}\ |\ n\in[3]\right\},\ \ \
 \mathcal{Z}_3=\left\{W_{n,2},W_{n,3}\ |\ n\in[3]\right\}.
 \end{align*}
 \item \textbf{Delivery Phase}: Assume that the request vector is $\mathbf{d}=(1,2,3)$. The server sends $W_{1,3}\bigoplus W_{2,2}$ $\bigoplus W_{3,1}$ to the users. Then each user can recover its requested file. For instance, user $1$ requests the file $W_1=\{W_{1,1},W_{1,2},W_{1,3}\}$ and has cached $W_{3,1}$, $W_{2,2}$, so it can recover $W_{1,3}$. The load is $R=\frac{1}{3}$.
 \end{itemize}
\hfill $\square$
\end{example}
For what said above, it follows that any PDA corresponds to a coded caching scheme achieving the performance state in the following lemma. 
\begin{lemma}\rm(\cite{YCTC})
\label{le-Fundamental}Given a $(K,F,Z,S)$ PDA, there  exists an $F$-division $(K,M,N)$ coded caching scheme with the memory ratio $\frac{M}{N}=\frac{Z}{F}$ and   load $R=\frac{S}{F}$.
\hfill $\square$
\end{lemma}

The authors in \cite{YCTC} showed that the seminal coded caching scheme proposed in~\cite{MN} can be represented by a special PDA, referred to as MN PDA.
\begin{construction}\rm (\emph{MN PDA}~\cite{YCTC})
	\label{con-MN} For any integer $t\in[K]$,  we have a $(t+1)$-$\left(K,{K\choose t},{K-1\choose t-1},{K\choose t+1}\right)$ PDA,
$\mathbf{P}=\left(\mathbf{P}(\mathcal{T},k)\right)$	 with dimension ${K\choose t}\times K$,   where  $\mathcal{T}\in {[K]\choose t}$ and $k\in [K]$, by\footnote{\label{foot:row and col indices of MN}Notice that the rows are indexed by all the subsets $\mathcal{T}\in {[K]\choose t}$, and the columns are indexed by all the  integers $k\in [K]$.}
	\begin{align}\label{Eqn_Def_AN}
		\mathbf{P}(\mathcal{T},k)=\left\{\begin{array}{ll}
			* & \ \ \  \mbox{if}~k\in\mathcal{T} \\
			\mathcal{T}\cup\{k\} & \ \ \  \mbox{otherwise}
		\end{array} \ .
		\right. 
	\end{align}
	\hfill $\square$
\end{construction}
When $K=3$ and $t=2$, the $3$-$(3,3,2,1)$ MN PDA is exactly the PDA in Example~\ref{MN-pda}.

In order to further reduce the subpacketization of the MN PDA,  the authors in  \cite{CJYT}  proposed a particular PDA construction, referred to as Partition PDA.
For the sake of clarity, we express the non-star entries by vectors first.
\begin{construction}\rm(\emph{Partition PDA}~\cite{CJYT})
	\label{con-general-1}
	For any positive integers $q$, $z$ and $m$ where $0<z<q$, we define the row index set as $\mathcal{F}=[q]^m$, and the column index set as $\mathcal{K}=[q]$. Then we have $m$ arrays $\mathbf{H}_1'$, $\ldots$, $\mathbf{H}_m'$, each of which is a $q^m\times q$ array.  For each $i\in[m]$, each entry in $\mathbf{H}_i'$   is defined as
	\begin{eqnarray}\label{Eqn_Gen._Par_PDA}
		\mathbf{H}_i'({\bf f},{k})=\left\{
		\begin{array}{ll}
			* & \ \ \  \textrm{if}\ k\in \mathcal{B}_{f_i}\\	
			(f_1,\ldots,f_{i-1}, k, f_{i+1}, \ldots,f_{m},<f_i-k>_q)
			&  \ \ \ \textrm{otherwise} 	
		\end{array} \ ,
		\right.
	\end{eqnarray}
	where ${\bf f}=(f_1,f_2,\ldots,f_{m})\in\mathcal{F}$  and $k\in\mathcal{K}$ represent the row and column indices, respectively, and
	$\mathcal{B}_{f_i}=\{f_i,<f_i+1>_q,\ldots,<f_i+(z-1)>_q\}$  represents the set of columns filled by ``*" in row $\mathbf{f}$.
	Intuitively, if $k\notin \mathcal{B}_{f_i}$, the entry $\mathbf{H}_i'({\bf f},{k})$ is a non-star entry  represented by a vector with length $m+1$, which is generated by replacing the $i^{\text{th}}$ coordinate of   ${\bf f}=(f_1,f_2,\ldots,f_{m})$ by   $k$  and  appending $<f_i-k>_q$ at the end of the vector.

	The block array formed by stacking  $\mathbf{H}_i'$  for all $i\in[m]$ next to each other as $\mathbf{H}'=(\mathbf{H}_1',\ldots,\mathbf{H}_m')$ is an  
	$m$-$(mq,q^m,zq^{m-1},q^{m}(q-z))$ PDA.
	\hfill $\square$
\end{construction}

For the sake of future convenience,
we replace the vectors in $\mathbf{H}'=(\mathbf{H}_1',\ldots,\mathbf{H}_m')$  by integers in $[q^{m}(q-z)]$ according to
an arbitrary   one-to-one mapping $\phi$;
the resulting array containing stars and integers is defined as $\mathbf{H}=(\mathbf{H}_1,\ldots,\mathbf{H}_m)$, which is also an
$m$-$(mq,q^m,zq^{m-1},q^{m}(q-z))$ PDA.

\begin{example}\rm
	\label{ex-par}
	When $q=3$, $z=2$ and $m=2$, the Partition PDA $\mathbf{H}'=(\mathbf{H}_1', \mathbf{H}_2')$ is illustrated in Fig~\ref{fig-par}, including two sub-arrays
	$\mathbf{H}_1'$ and $\mathbf{H}_2'$ with dimension
	$9\times3$.
	\begin{figure}
		\centering
		\includegraphics[width=4.3in]{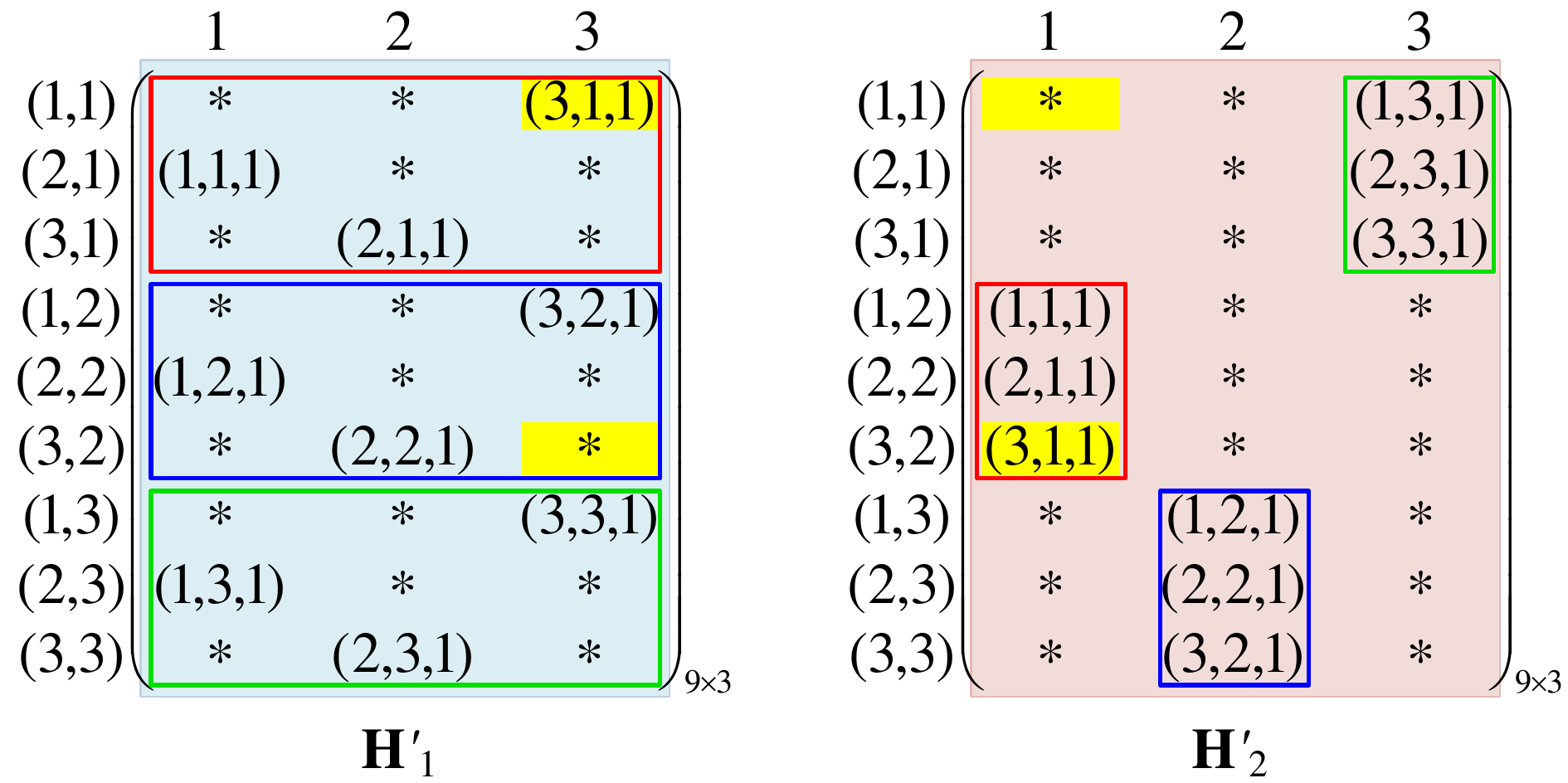}
		\caption{$\mathbf{H}_1'$, $\mathbf{H}_2'$ of Partition PDA $\mathbf{H}'$ with $q=3$, $z=2$ and $m=2$}\label{fig-par}
	\end{figure}
	In $\mathbf{H}_1'$, let us focus
	on the row with index  $\mathbf{f}=(1,1)$,
	\begin{itemize}
		\item  we have $\mathcal{B}_{f_1}=\{f_1,<f_1+1>_3\}=\{1,2\}$. Thus the stars are located at column $1$ and column $2$ of this row, i.e., $\mathbf{H}_1'({\bf f},1)=\mathbf{H}_1'({\bf f},2)=*$;
		\item for column $3$ of this row (i.e., when $k=3$), we have $\mathbf{H}_1'(\mathbf{f},3)=(k,f_2,<f_1-k>_3)=(3,1,1)$ from \eqref{Eqn_Gen._Par_PDA}.
	\end{itemize}
	In $\mathbf{H}_2'$, let us focus
	on the row with index  $\mathbf{f}=(2,1)$,
	\begin{itemize}
		\item  we have $\mathcal{B}_{f_2}=\{f_2,<f_2+1>_3\}=\{1,2\}$. Thus the  stars are located at column $1$ and column $2$ of   this row, i.e., $\mathbf{H}_2'({\bf f},1)=\mathbf{H}_2'({\bf f},2)=*$;
		\item for column $3$ of this row (i.e., when $k=3$), we have $\mathbf{H}_2'(\mathbf{f},3)=(f_1,k,<f_2-k>_3)=(2,3,1)$ from \eqref{Eqn_Gen._Par_PDA}.
	\end{itemize}
	Similarly,  the other entries in $\mathbf{H}_1'$ and $\mathbf{H}_2'$ are obtained as illustrated in Fig~\ref{fig-par}. Next, we check the Condition C3 of PDA in  Definition~\ref{def-PDA}.   For instance, let us focus on the vector $(3,1,1)$. In $\mathbf{H}_1'$, the vector $(3,1,1)$ is filled in the entry at row $(1,1)$ and column $3$; in $\mathbf{H}_2'$, the vector $(3,1,1)$ is filled in the entry at row $(3,2)$ and column $1$. In row $(1,1)$, the entry at column $1$ of $\mathbf{H}_2'$ is star; in row $(3,2)$, the entry at column $3$ of $\mathbf{H}_1'$ is star. Thus the sub-array containing $(3,1,1)$
	satisfies  Condition C3 of PDA in Definition~\ref{def-PDA}.
	
	We can also replace the vectors in $\mathbf{H}_1'$ and $\mathbf{H}_2'$ by integers according to the one-to-one mapping $\phi$   in Table~\ref{tab-mapping}.
	\begin{table}
		\center
		\caption{The mapping $\phi$ }\label{tab-mapping}
		\begin{tabular}{|c|c|c|c|c|c|c|c|c|c|}
			\hline
			$\mathbf{h}'$ & $(1,1,1)$ & $(2,1,1)$ & $(3,1,1)$ & $(1,2,1)$ & $(2,2,1)$ & $(3,2,1)$ & $(1,3,1)$ & $(2,3,1)$ & $(3,3,1)$ \\ \hline
			$\mathbf{h}=\phi(\mathbf{h}')$ & $1$ & $2$ & $3$ & $4$ & $5$ & $6$ & $7$ & $8$ & $9$ \\ \hline
		\end{tabular}
	\end{table}
	The resulting arrays $\mathbf{H}_1 $ and $\mathbf{H}_2 $
	are illustrated in Fig.~\ref{fig-par-inte}.
	\begin{figure}
	\centering
	\includegraphics[width=2.3in]{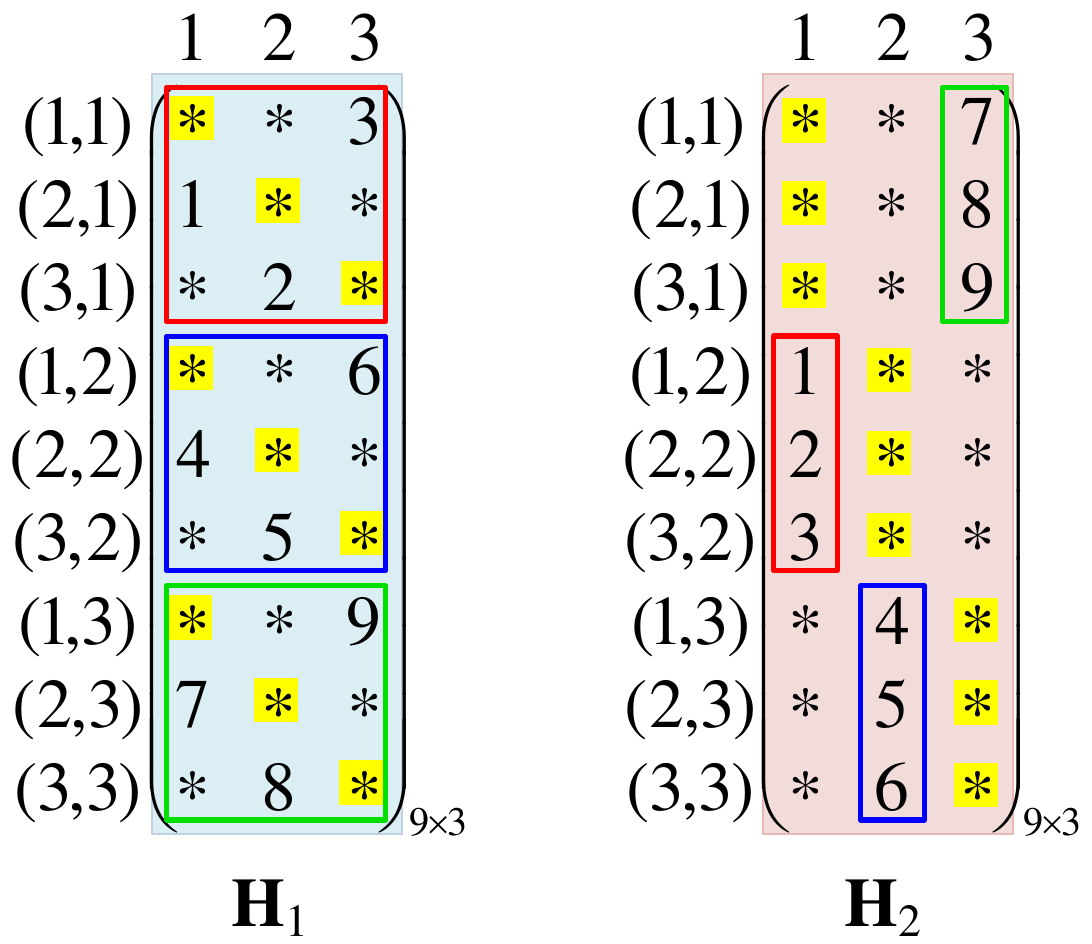}
	\caption{$\mathbf{H}_1$, $\mathbf{H}_2$ of Partition PDA $\mathbf{H}$ with $q=3$, $z=2$ and $m=2$}\label{fig-par-inte}
\end{figure}
	\hfill $\square$
\end{example}
Next, we introduce the concept of ``tag-star" in a Partition PDA, which plays an important role in the constructions of this paper.
\begin{definition}\rm 
	\label{rem-par}
	For    each $i\in[m]$ and 	
	each $\mathbf{f}\in\mathcal{F}$, we have $\mathbf{H}_i(\mathbf{f},f_i)=*$. Then we define this star (i.e.,
	the star located at  row $\mathbf{f}$ and column $f_i $ of $\mathbf{H}_i$)  as a {\it tag-star}.		
	\hfill $\square$
\end{definition}

In Example~\ref{ex-par}, for $i=1$, the row indexed by $\mathbf{f} = (2,1)$ of $\mathbf{H}_1$ (see Fig.~\ref{fig-par-inte}) contains a $*$ in the second position. Since $f_1 = 2$, this is a tag star. 
For $i = 2$, the row indexed by $\mathbf{f} = (2,1)$ of $\mathbf{H}_2$ (see Fig.~\ref{fig-par-inte}) contains a $*$ in the first position. Since $f_2 = 1$, this is a tag star.

\subsection{1D Multi-access Coded Caching Model}
\label{subsect-line}
A $(K,L,M,N)$ 1D MACC system proposed in \cite{JHNS} contains a server with a set of $N$ equal-length files, $K$ cache-nodes, and $K\leq N$ cache-less users. Each cache-node has a memory size of $M$ files where $0\leq M \leq \frac{N}{L}$. The cache-nodes are placed in a line, and each user can access $L$ neighboring cache-nodes in a cyclic wrap-around fashion.
That is, each user $U_k$, $k\in [K]$, can retrieve all the content cached by the cache-node $C_{k'}$ if and only if $<k-k'>_K <L$. Each user is also connected to the server via an error-free shared link.
As in~\cite{JHNS}, we assume that the users can retrieve the cache content of the connected cache-nodes without any cost.
The system operates in two phases.
\begin{itemize}
  \item {\bf Placement phase:} Each file is divided into $F$ packets of equal size, then each cache-node directly caches up to $MF$ packets of files.
  Each user $U_k$ can retrieve the content stored at its accessible cache-nodes.
  The placement phase is done without knowledge of later requests.
  \item {\bf Delivery phase:} Each user randomly requests one file. According to the request vector $ \mathbf{d}=(d_{1},d_{2},\ldots,d_{K})$ and the retrieved content by users, the server transmits $S_{{\bf d}}$ multicast messages to users, such that each user's request can be satisfied.
\end{itemize}

Let $t=\frac{KM}{N} \in [0:\left\lfloor K/L \right\rfloor ]$.
A transformation approach to generate a  $(K,L,M,N)$ 1D MACC scheme was proposed in  \cite{CWLZC} which extends any $(K'=K-t(L-1),F',Z',S')$ PDA  $\mathbf{P}$ for the shared-link coded caching system satisfying  $\frac{K'Z'}{F'}=t$,  Conditions C$1$-C$3$ in Definition~\ref{def-PDA} and    Conditions C$4$, C$5$ (which will be clarified soon).

First, $\mathbf{P}$ satisfies the following condition,
\begin{itemize}
	\item C4. Each row of $\mathbf{P}$ has exactly $t$ stars.
\end{itemize}
Then we define
\begin{align}
	\label{eq-mathcal-A_j}
	\mathcal{A}_j=\{k' \ | \ \mathbf{P}(j,k')=*,k'\in[K'],  j\in[F']\}
\end{align}
as column index set of star entries in row $j$ of $\mathbf{P}$. Notice that $|\mathcal{A}_j|=t$ since each packet is cached $t$ times in the original MN caching model.

 In $(K,L,M,N)$ 1D MACC scheme, each file is divided into $K$ equal-length subfiles, $W_n=\{W_n^{(r)} | r\in[K]\}$ where $n\in[N]$, and the caching procedure is also divided into $K$ rounds. For any $r\in[K]$, in the $r^{\text{th}}$ round we only deal with the $r^{\text{th}}$ subfile of each file. Furthermore, it is sufficient to introduce the construction in the first round since all the caching procedures in different rounds are symmetric.
In the first round, the authors in \cite{CWLZC} showed that the node-placement array $\mathbf{C}_{\text{1D}}$,  user-retrieve array $\mathbf{U}_{\text{1D}}$, and  user-delivery array $\mathbf{Q}_{\text{1D}}$ are generated by the PDA $\mathbf{P}$
as follows.
\begin{itemize}
	\item \textbf{Node-placement array} $\mathbf{C}_{\text{1D}}$. In order to obtain the maximum local caching gain, the scheme guarantees that any $L$ neighboring cache-nodes do not cache any common packets. Thus based on $\mathbf{P}$, the $F'\times K$ node-placement array
	$\mathbf{C}_{\text{1D}}=\left(\mathbf{C}_{\text{1D}}(j, k) \right)_{j\in [F'],k\in [K]}$ is defined as
	\begin{align}
		\label{eq-array-node-caching}
		\mathbf{C}_{\text{1D}}(j, k)=\left\{
		\begin{array}{ll}
			* & \ \ \ \hbox{if}\ \ k\in \mathcal{C}_{j} \\
			\text{null} & \ \ \ \hbox{otherwise}
		\end{array} \ ,
		\right.
	\end{align}
	where
	\begin{align}
		\mathcal{C}_{j}=\left\{\mathcal{A}_j[i]+(i-1)(L-1)\ | \ i\in[t] \right\}. \label{eq-caching-index}
	\end{align}
	Here  $\mathcal{C}_{j}$ represents the set of cache-nodes caching the packet indexed by $j$. From \eqref{eq-caching-index}, there are $t$ stars in each row, which means that each packet stored by $t$ cache-nodes. Moreover, any two entries $k_1=\mathcal{C}_{j}[i_1]$ and  $k_2=\mathcal{C}_{j}[i_2]$ in $\mathcal{C}_{j}$ satisfy $D_{\rm r}(k_1,k_2)\geq L$,\footnote{Recall that $D_{\rm r}(k_1,k_2)=\min\{<k_1-k_2>_{K},K-<k_1-k_2>_{K}\}$ where $k_1,k_2\in[K]$.} which means that any two cache-nodes accessed by the same users do not cache common packets.
	\item \textbf{User-retrieve array} $\mathbf{U}_{\text{1D}}$. According to the relationship between users and their accessible cache-nodes, based on $\mathbf{C}_{\text{1D}}$, the $F'\times K$ user-retrieve array $\mathbf{U}_{\text{1D}}=\left(\mathbf{U}_{\text{1D}}(j, k) \right)_{j\in [F'],k\in [K]}$ is defined as
	\begin{align}
		\label{eq-array-user-caching}
		\mathbf{U}_{\text{1D}}(j, k)=\left\{
		\begin{array}{ll}
			* & \ \ \ \hbox{if}\ \ k\in \mathcal{U}_{j} \\
			\text{null} & \ \ \ \hbox{otherwise}
		\end{array} \ ,
		\right.
	\end{align}where
	\begin{align}
		\mathcal{U}_{j} =\bigcup\limits_{i\in[t]}\big\{\mathcal{C}_{j}[i], \mathcal{C}_{j}[i]+1, \ldots, \mathcal{C}_{j}[i]+(L-1)\big\}.
		\label{eq-user-index}
	\end{align}
	Here  $\mathcal{U}_{j}$ represents the set of users who can retrieve the packet indexed by $j$. From \eqref{eq-caching-index} and \eqref{eq-user-index}, there are $tL$ stars in each row, which means that each packet can be retrieved by $tL$ users. In addition, for any row indexed by $j\in[F']$, we can divide the $tL$ starts into $t$ groups and each group has $L$ consecutive stars, where the $i^{\text{th}}$ group is defined as
	\begin{align}
		\mathcal{U}_{j,i} := \big\{\mathcal{C}_{j}[i], \mathcal{C}_{j}[i]+1, \ldots, \mathcal{C}_{j}[i]+(L-1)\big\}.
		\label{eq:star group}
	\end{align}
	$\mathcal{U}_{j,i}$ represents the set of users who retrieve the packet cached by cache-node $\mathcal{C}_{j}[i]$.
	
	Another important observation is that, each packet cannot be retrieved by $K-tL=K'-t$ users, which is exactly the same as the number of users not caching each packet in the original $(K',F',Z',S')$ PDA $\mathbf{P}$.
	
	\item \textbf{User-delivery array} $\mathbf{Q}_{\text{1D}}$. Based on  $\mathbf{U}_{\text{1D}}$ and $\mathbf{P}$, the $F'\times K$ user-delivery array $\mathbf{Q}_{\text{1D}}=\left(\mathbf{Q}_{\text{1D}}(j, k) \right)_{j\in[F'] ,k\in [K]}$ is defined as
	\begin{align}
		\label{eq-array-user-PDA}
		\mathbf{Q}_{\text{1D}}(j, k)=\left\{
		\begin{array}{ll}
			s & \ \ \ \hbox{if}\ \ k\in\overline{\mathcal{U}}_{j} \\
			* & \ \ \ \hbox{otherwise}
		\end{array} \ ,
		\right.
	\end{align}
	where $s=\mathbf{P}(j,\psi_{j}(k))\in[S']$. $\psi_{j}$ is a one-to-one mapping from $\overline{\mathcal{U}}_{j}=[K]\setminus\mathcal{U}_{j}$ (i.e., the column index set of $\mathbf{U}_{\text{1D}}$ where the entries in row $j$ are non-star)  to $\overline{\mathcal{A}_j}=[K']\setminus \mathcal{A}_j$ (i.e., the column index set of original PDA $\mathbf{P}$ where the entries in row $j$ are non-star).
	So,
	\begin{align}
		\psi_{j}(\overline{\mathcal{U}}_{j} [\mu]):=\overline{\mathcal{A}_j}[\mu], \ \forall \mu\in [K'-t], j\in F'.
		\label{eq:definition of psi}
	\end{align}
	Hence, the alphabet set of the resulting $\mathbf{Q}_{\text{1D}}$ contains $S'$ different integers which indicate the broadcasted messages.
\end{itemize}

\begin{remark}\rm
	\label{remark-C5}
	In \cite{CWLZC}, the authors showed that in order to guarantee $\mathbf{Q}_{\text{1D}}$ satisfying Condition C3, the original PDA $\mathbf{P}$ should satisfy,
	\begin{itemize}
		\item C$5$. For any two distinct entries $\mathbf{P}(j_1,k_1)$ and $\mathbf{P}(j_2,k_2)$ satisfying  $\mathbf{P}(j_1,k_1)=\mathbf{P}(j_2,k_2)=s\in[S']$, assume that $k_1=\mathcal{A}_{j_1}\bigcup\{k_1\}[i_1]$, $k_2=\mathcal{A}_{j_2}\bigcup\{k_2\}[i_2]$, we have $k_1+(i_1-1)(L-1)\in\mathcal{U}_{j_2}$ and $k_2+(i_2-1)(L-1)\in\mathcal{U}_{j_1}$ for some integers $i_1,i_2\in[t+1]$.
	\end{itemize}
	\hfill $\square$
\end{remark}
It turns out that most existing PDAs meet Conditions C4 and C5, such as  MN PDA and the PDAs in \cite{YCTC,CJYT,SZG,YTCC}.

After determining $\mathbf{C}_{\text{1D}}$, $\mathbf{U}_{\text{1D}}$ and $\mathbf{Q}_{\text{1D}}$, the placement and delivery strategies of $(K,L,M,N)$  1D MACC scheme are obtained in the first round.  For each $r\in[K]$, in the $r^{\text{th}}$ round, we only need to cyclically right-shift $\mathbf{C}_{\text{1D}}$, $\mathbf{U}_{\text{1D}}$, and $\mathbf{Q}_{\text{1D}}$ by $r-1$ positions,  respectively. For the total placement array of cache-nodes, there are $K'Z'$ stars in each column,  while this array has $KF'$ rows. Hence, each cache-node caches $\frac{K'Z'}{F'K}N=\frac{KM}{N}\cdot\frac{N}{K}=M$ files, satisfying the memory size constraint.

In conclusion, by the above transformation approach, we obtain the following theorem.
\begin{theorem}\rm(\cite{CWLZC})
\label{lem-1D-tran}
Given any $(K',F',Z',S')$ PDA $\mathbf{P}$ satisfying Conditions C$1$-C$5$,  there exists a $(K=K'+t(L-1),L,M,N)$ 1D MACC scheme where $t=\frac{KM}{N}=\frac{K'Z'}{F'}$, with the
transmission load $R_{\text{1D}}=\frac{S'}{F'}$.
\hfill $\square$
\end{theorem}


By applying the transformation approach into the MN PDA, a 1D MACC scheme can be obtained as follows,  referred to as CWLZC 1D MACC scheme.
\begin{corollary}\rm(\emph{CWLZC 1D MACC scheme})
	\label{lem-line}
	Given a $\left({K'},{K' \choose t},{K'-1 \choose t-1},{K' \choose t+1}\right)$ MN PDA $\mathbf{P}$, there exists a 1D MACC scheme for $(K=K'+t(L-1),L,M,N)$ 1D MACC system where $t=\frac{KM}{N}=\frac{K'Z'}{F'}$, with the transmission load $R_{\text{1D}}=\frac{K-tL}{t+1}$
	\hfill $\square$
\end{corollary}

\begin{example}\rm
	\label{ex-1D-CWLZC}
	We consider a $(K,L,M,N)=(5,2,6,15)$ 1D MACC scheme based on the $(K',F',Z',S')=(3,3,2,1)$ MN PDA $\mathbf{P}$. The construction of node-placement array $\mathbf{C}_{\text{v}}$, user-retrieve array $\mathbf{U}_{\text{v}}$, and user-delivery array $\mathbf{Q}_{\text{v}}$ in the first round is illustrated in Fig.~\ref{fig-line-macc}.
	\begin{figure}
		\centering
		\includegraphics[width=5in]{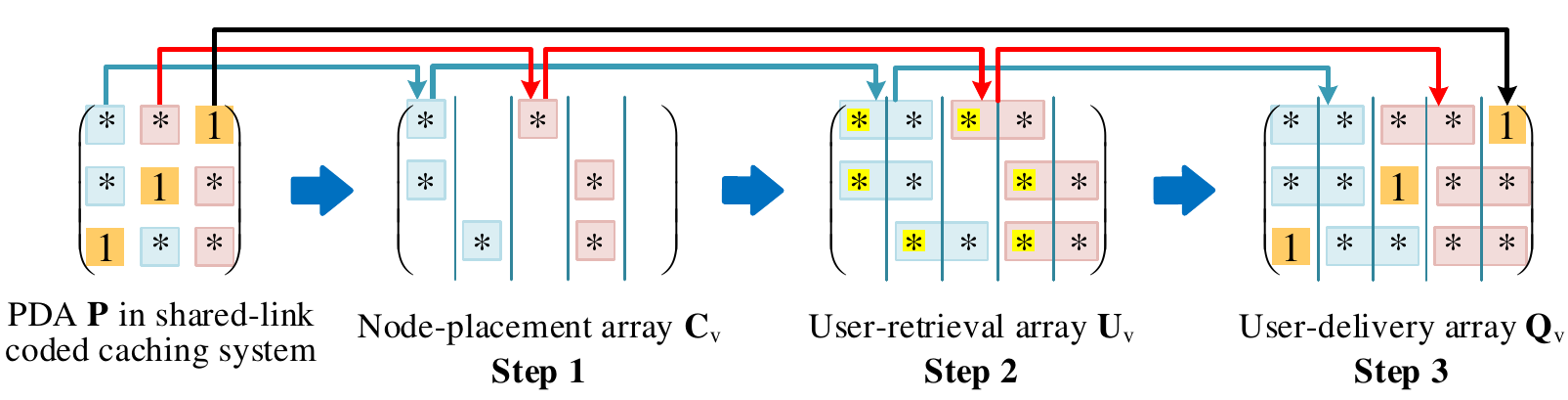}
		\caption{Transformation approach for MN PDA $\mathbf{P}$ to generate 1D MACC scheme}\label{fig-line-macc}
	\end{figure}
	\begin{itemize}
		\item The node-placement array $\mathbf{C}_{\text{v}}$ is constructed from $\mathbf{P}$. More precisely, in row $1$, the stars at columns $1$ and $2$ of $\mathbf{P}$ are corresponded to columns $1+(1-1)(2-1)=1$ and $2+(2-1)(2-1)=3$ of $\mathbf{C}_{\text{v}}$, respectively; in row $2$, the stars at columns $1$ and $3$ of $\mathbf{P}$ are corresponded to columns $1+(1-1)(2-1)=1$ and $3+(2-1)(2-1)=4$ of $\mathbf{C}_{\text{v}}$, respectively; in row $3$, the stars at columns $2$ and $3$ of $\mathbf{P}$ are corresponded to columns $2+(1-1)(2-1)=2$ and $3+(2-1)(2-1)=4$ of $\mathbf{C}_{\text{v}}$, respectively. By this construction, any two cache-nodes connected to the same users do not cache any common packets.
		\item The user-retrieve array $\mathbf{U}_{\text{v}}$ is constructed from $\mathbf{C}_{\text{v}}$ by the relationship between users and the accessible cache-nodes. More precisely, in row $1$, the first star at column $1$ of $\mathbf{C}_{\text{v}}$ is extended to the stars at columns $1$ and $2$ in the first group of $\mathbf{U}_{\text{v}}$, the second star at column $3$ of $\mathbf{C}_{\text{v}}$ is extended to the stars at columns $3$ and $4$ in the second group of $\mathbf{U}_{\text{v}}$; in row $2$, the first star at column $1$ of $\mathbf{C}_{\text{v}}$ is extended to the stars at columns $1$ and $2$ in the first group of $\mathbf{U}_{\text{v}}$, the second star at column $4$ of $\mathbf{C}_{\text{v}}$ is extended to the stars at columns $4$ and $5$ in the second group of $\mathbf{U}_{\text{v}}$;   in row $3$, the first star at column $2$ of $\mathbf{C}_{\text{v}}$ is extended to the stars at columns $2$ and $3$ in the first group of $\mathbf{U}_{\text{v}}$, the second star at column $4$ of $\mathbf{C}_{\text{v}}$ is extended to the stars at columns $4$ and $5$ in the second group of $\mathbf{U}_{\text{v}}$. Thus each packet can be retrieved by $tL=4$ users.
		\item The user-delivery array $\mathbf{Q}_{\text{v}}$ is constructed by filling the non-star entries in $\mathbf{U}_{\text{v}}$ according to the integers in $\mathbf{P}$. More precisely, in row $1$, the integer ``$1$" at column $3$ of $\mathbf{P}$ is filled at column $5$ of $\mathbf{Q}_{\text{v}}$; in row $2$, the integer ``$1$" at column $2$ of $\mathbf{P}$ is filled at column $3$ of $\mathbf{Q}_{\text{v}}$; in row $3$, the integer ``$1$" at column $1$ of $\mathbf{P}$ is filled at column $1$ of $\mathbf{Q}_{\text{v}}$.
	\end{itemize}
	After determining $\mathbf{C}_{\text{v}}$, $\mathbf{U}_{\text{v}}$, and $\mathbf{Q}_{\text{v}}$, we have the placement and delivery strategies for the $(K_1,L,M_1,N)=(5,2,6,15)$ 1D MACC system in the first round. Assume that the request vector is $\mathbf{d}=(1,2,\ldots,5)$, the server sends $W_{1,3}^{(1)}\oplus W_{3,2}^{(1)}\oplus W_{5,1}^{(1)}$ to the users. Then the overall transmission load is $R_{\text{1D}}=\frac{1\times5}{3\times5}=\frac{1}{3}$ which coincides with Corollary~\ref{lem-line}.
	\hfill $\square$
\end{example}

Another class of useful 1D MACC schemes for this paper are generated by using the transformation approach \cite{CWLZC} on the Partition PDA proposed in \cite{CJYT}.

\begin{corollary}\rm(\emph{1D MACC scheme based on Partition PDA})
	\label{lem-line-par}
	Given a $m$-$(mq,q^m,zq^{m-1},q^m(q-z))$ Partition PDA $\mathbf{H}=(\mathbf{H}_1,\mathbf{H}_2,\ldots,\mathbf{H}_m)$, there exists $m$ different 1D MACC schemes for  $(K=q,L=z,M=N/K,N)$ 1D MACC system.
	The overall transmission load for $m$ schemes is $R_{\text{1D}}=K-L$.
	\hfill $\square$
\end{corollary}

\begin{example}\rm
\label{ex-par-1d-macc}
We consider two different $(K,L,M,N)=(3,2,5,15)$ 1D MACC schemes generated by using the transformation approach on $2$-$(6,9,6,9)$ Partition PDA $\mathbf{H}=(\mathbf{H}_1,\mathbf{H}_2)$. The node-placement arrays $\mathbf{E}_1$, $\mathbf{E}_2$, user-retrieve arrays $\mathbf{B}_1$, $\mathbf{B}_2$, and user-delivery arrays $\mathbf{H}_1$, $\mathbf{H}_2$ are illustrated in Fig.~\ref{fig-par-macc}.
\begin{itemize}
	\item The node-placement arrays $\mathbf{E}_1$, $\mathbf{E}_2$ consist of the \textit{tag-stars} (defined in Definition~\ref{rem-par}) in sub-Partition arrays $\mathbf{H}_1$, $\mathbf{H}_2$. Recall that, the row indices exactly indicate the positions of tag-stars where the first coordinate corresponds to the column index of stars in $\mathbf{E}_1$; the second coordinate corresponds to the column index of stars in $\mathbf{E}_2$.
	\item The user-retrieve arrays  $\mathbf{B}_1$, $\mathbf{B}_2$ are constructed from $\mathbf{E}_1$, $\mathbf{E}_2$ by the relationship between users and the accessible cache-nodes. More preciously, $\mathbf{B}_1$, $\mathbf{B}_2$ are generated by extending each star in $\mathbf{E}_1$, $\mathbf{E}_2$  to $L=2$ stars.
	For instance, in the row indexed by $(2,1)$, the star at column $2$ of $\mathbf{E}_1$ is extended to the stars at columns $2$ and $3$  of $\mathbf{B}_1$; the star at column $1$ of $\mathbf{E}_2$ is extended to the stars at columns $1$ and $2$  of $\mathbf{B}_2$.
	Notice that, $\mathbf{B}_1$, $\mathbf{B}_2$ have the same star entries as $\mathbf{H}_1$, $\mathbf{H}_2$ of Partition PDA $\mathbf{H}=(\mathbf{H}_1,\mathbf{H}_2)$ in Example~\ref{ex-par}.
	\item The user-delivery arrays  $\mathbf{H}_1$, $\mathbf{H}_2$ are constructed by filling the non-star entries in $\mathbf{B}_1$, $\mathbf{B}_2$.
	Since $\mathbf{B}_1$, $\mathbf{B}_2$ have the same star entries as sub-Partition PDAs $\mathbf{H}_1$, $\mathbf{H}_2$ in Fig~\ref{fig-par-inte}, $\mathbf{H}_1$, $\mathbf{H}_2$ are used to be user-delivery arrays.
\end{itemize}
\begin{figure}
	\centering
	\includegraphics[width=5.5in]{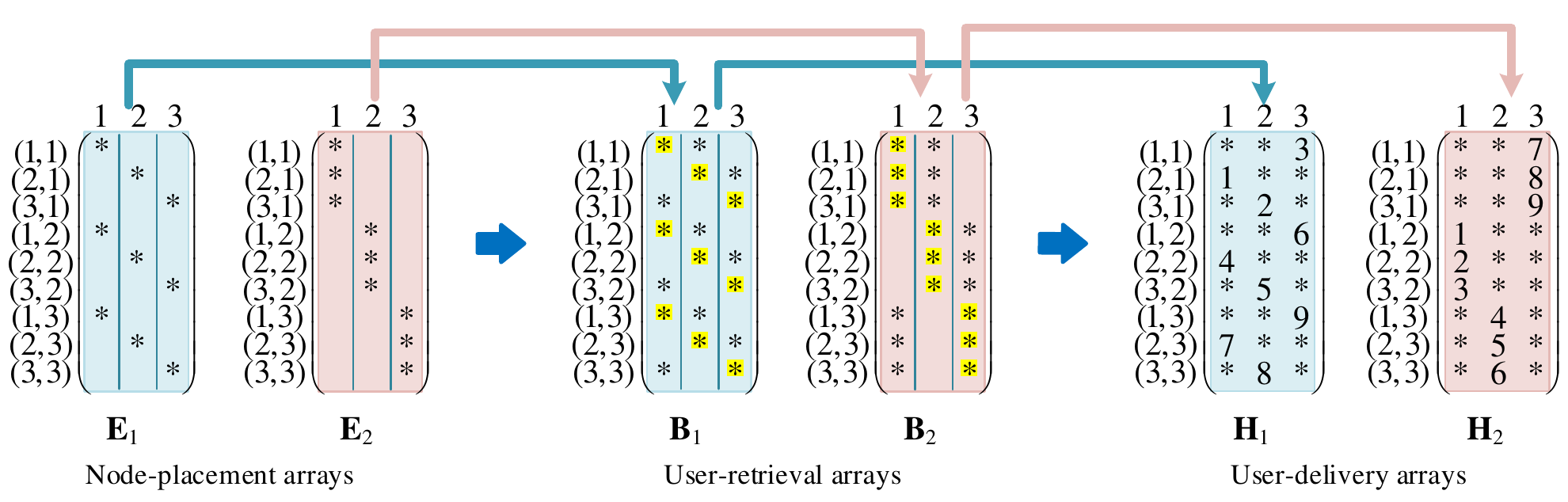}
	\caption{$(3,2,5,15)$ 1D MACC schemes generated by Partition PDA $\mathbf{H}=(\mathbf{H}_1,\mathbf{H}_2)$}\label{fig-par-macc}
\end{figure}\hfill $\square$
\end{example}

\section{System Model: 2D Multi-access Coded Caching Model}
\label{sect-system}
The new 2D MACC model considered in this paper, with parameters $(K_1,K_2,L,M,N)$, is given as follows. 
A server containing $N$ equal-length files is connected to  $K:=K_1\times K_2$  cache-less users through an error-free shared link. Without loss of generality, we assume that $K_1\geq K_2$. There are
also  $K$ cache-nodes, each of which  has a memory size of $M$ files where $0\leq M \leq \frac{N}{L^2}$.
The cache-nodes are placed in a $K_1\times K_2$ array (see Fig.~\ref{fig-grid-intru}), and at the position of each cache-node there is one user.
For any positive integers $k_1\in [K_1]$ and $k_2\in [K_2]$, the cache-node and the user   located at   row $k_1$ and column $k_2$ of the $K_1\times K_2$ array, are denoted by
 $C_{k_1,k_2}$ and $U_{k_1,k_2}$, respectively.
Each user $U_{k_1,k_2}$ can retrieve all the content cached by the cache-node $C_{k_1',k_2'}$,
  if and only if  the modular distances $<k_1-k_1'>_{K_1}$ and $<k_2-k_2'>_{K_2}$
%
are less than  $L$, i.e.,
\begin{eqnarray}
\label{eq-distance-grid}
\max\{<k_1-k_1'>_{K_1}, <k_2-k_2'>_{K_2} \}<L.
\end{eqnarray}
Each user is connected to $L^2$ neighboring cache-nodes in a cyclic wrap-around fashion. For instance, in Fig.~\ref{fig-grid-intru}  user $U_{2,2}$  can access $L^2=4$ neighboring cache-nodes $C_{1,1}$, $C_{1,2}$, $C_{2,1}$, $C_{2,2}$.
Similar with 1D MACC model, we assume that the users can retrieve the cache content of the connected cache-nodes without any cost.
A $(K_1,K_2,L,M,N)$ 2D MACC scheme consists of two phases,
\begin{itemize}
  \item {\bf Placement phase:}
  The server divides each file into $F$ packets with equal size. 
  For any positive integers $k_1\in [K_1]$ and $k_2\in[K_2]$, cache-node $C_{k_1,k_2}$ caches up to $MF$ packets of files. Each user $U_{k_1,k_2}$ can retrieve the packets stored by its connected cache-nodes. This phase  is  done  without knowledge of later requests.
  \item {\bf Delivery phase:} For any request vector $\mathbf{d}=(d_{1,1}, \ldots, d_{K_1,K_2})$ representing that user $U_{k_1,k_2}$ requests file $W_{d_{k_1,k_2}}$ where $k_1\in [K_1]$ and $k_2\in [K_2]$,  the server transmits $S_{{\bf d}}$ coded packets to users such that each user can decode its requested file.
\end{itemize}
Notice that, each column (or row) in the 2D MACC model is exactly a 1D MACC model in Section~\ref{subsect-line}.
The objective of the 2D MACC problem is to minimize the worst-case load
$R=\max\left\{\frac{ S_{\mathbf{d}}}{F}\ \Big|\ \mathbf{d}\in[N]^K\right\}$,
as defined in~\eqref{eq:def of load}.


\section{Main Results}
\label{sec-main result}

In this section, our new schemes for the 2D MACC network are presented. We first provide a baseline scheme by directly extending the CWLZC 1D MACC scheme. Then,  when $L|K_1$ and $L|K_2$,  we improve the baseline scheme by a grouping scheme. Next, for the more general case where $L<\min\{K_1,K_2\}$ (i.e., $L<K_2$), we propose a new transformation approach, which  constructs  a hybrid  2D MACC scheme from two classes of 1D MACC schemes.
 Thus, concatenating the new transformation approach with the transformation approach from   PDAs to 1D MACC schemes proposed in~\cite{CWLZC},
we obtain a transformation approach from PDAs to 2D MACC schemes.
 Finally, the performance analysis and construction examples of these two schemes are introduced.

\subsection{Proposed 2D MACC Schemes}
\label{subsec-main}
By directly using the CWLZC 1D MACC scheme in Corollary~\ref{lem-line} into the 2D model, we obtain the following baseline scheme.
\begin{theorem}\rm(\emph{Baseline Scheme})
\label{th-baseline}
For the $(K_1,K_2,L,M,N)$ 2D MACC problem, the lower convex envelope of the following memory-load tradeoff corner points is achievable,
\begin{itemize}
\item when $K_2\leq L$, \begin{align}
(M, R_1)= \left(\frac{N t}{K_1K_2}, \frac{K_1K_2-tLK_2}{t+1} \right), \ t\in\left[0: \left\lfloor\frac{K_1}{L}\right\rfloor\right], \label{eq:load baseline1}
\end{align}
and $(M,R_1)=\left(\frac{N}{K_2L},0\right)$;
\item when $K_2> L$,
\begin{align}
(M, R_1)= \left(\frac{N t}{K_1K_2}, \frac{K_1K_2-tL^2}{\gamma t+1} \right), \ t\in\left\{0,\frac{1}{\gamma},\ldots,\left\lfloor\frac{K_1}{L}\right\rfloor\frac{1}{\gamma}\right\},  \label{eq:load baseline2}
\end{align}%
where $\gamma=\frac{L}{K_2}$, and $(M,R_1)=\left(\frac{N}{L^2},0\right)$.
\end{itemize}
\hfill $\square$
\end{theorem}

\begin{proof}
	Assume that in a $(K_1,K_2,L,M,N)$ 2D MACC system, each file has $F$ packets where $L$ and $K_2$ divide $F$, and each packet has enough bits such that any field extension can be operated. 
	The $K_1\times K_2$ 2D topology can be divided into $K_2$ columns, such that each column can be regarded as a 1D MACC system with $K_1$ cache-nodes and users, where each user can access $L$ neighboring  cache-nodes.
	
	We first consider the case  $K_2\leq L$.
	\begin{itemize}
		\item \textbf{Placement Phase}. We divide each file into $K_2$ subfiles with equal length, i.e., $W_n=\Big\{W^{(k_2)}_{n} \ |$ $ k_2\in[K_2]\Big\}$ where $n\in[N]$. Each subfile has $\frac{F}{K_2}$ packets. For each integer $k_2\in[K_2]$, denote the set of all the $k_2^{\text{th}}$ subfiles by $\mathcal{W}^{(k_2)}=\{W^{(k_2)}_{n}\ |\  n\in[N]\}$. Then, the server places $\mathcal{W}^{(k_2)}$ into the cache-nodes in the $k^{\text{th}}_2$ column by using the placement strategy of $(K_1,L,M_1=K_2M,N)$ CWLZC 1D MACC scheme in Corollary~\ref{lem-line}. Each cache-node totally caches $M_1\cdot\frac{F}{K_2}=K_2M\cdot\frac{F}{K_2}=MF$ packets.
		
		\item \textbf{Delivery Phase}. Given any demand vector $\mathbf{d}$, the server sends the coded subfiles of $\mathcal{W}^{(k_2)}$ to the users in the $k^{\text{th}}_2$ column by using the delivery strategy of $(K_1,L,M_1,N)$ CWLZC 1D MACC scheme, for each $k_2\in [K_2]$.  Since the server uses the delivery strategy of $(K_1,L,M_1,N)$ CWLZC  scheme exactly $K^2_2$ times, from Corollary~\ref{lem-line}, the transmission load is
		\begin{eqnarray*}
			R_1=K_2^2\cdot\frac{K_1(1-\frac{M_1L}{N})}{\frac{K_1 M_1}{N}+1}\cdot \frac{1}{K_2}=\frac{K_1K_2-tLK_2}{t+1},
		\end{eqnarray*} where $ t=\frac{K_1M_1}{N} \in \left[0: \left\lfloor\frac{K_1}{L}\right\rfloor\right]$.
		
		\item \textbf{Decodability}. In the 2D MACC system, each user can access all the $K_2$ cache-nodes in each row which cache $K_2$ different subfiles. Hence, each user can totally obtain $K_2$ subfiles of each file from the placement and delivery phases, such that it can decode its desired file.
	\end{itemize}
	
	Similar to the above case, the scheme for the case $K_2> L$ can be obtained as follows.  In the placement phase, each file $W_n$ where $n\in[N]$ is divided into $L$ non-overlapping and equal-length subfiles, which are then encoded into $K_2$ subfiles by a $[K_2,L]$ MDS code, i.e., $\widetilde{W}_n=\left\{\widetilde{W}^{(k_2)}_{n}\ |\ k_2\in[K_2]\right\}$. Each MDS-coded subfile has $\frac{F}{L}$ packets. The server places $\{\widetilde{W}^{(k_2)}_n\ |\ n\in[N]\}$ to the cache-nodes in the $k^{\text{th}}_2$ column by using the $(K_1,L,M_1'=LM,N)$ CWLZC 1D MACC scheme.  Each cache-node caches $M_1'\cdot\frac{F}{L}=LM\cdot\frac{F}{L}=MB$ packets, satisfying the memory size constraint.  In the delivery phase, for each $l\in[L]$, the server sends the required subfiles of $\{W^{(l)}_{n}\ |\  n\in[N]\}$ to the users in each column by using the  $(K_1,L,M_1',N)$ CWLZC 1D MACC scheme.  Since the server uses the delivery strategy of $(K_1,L,M_1',N)$ CWLZC scheme exactly $K_2L$ times, the transmission load is
	\begin{eqnarray*}
		R_1&=&K_2L\cdot\frac{K_1(1-\frac{L^2M}{N})}{\frac{K_1LM}{N}+1}\cdot \frac{1}{L}=\frac{K_1K_2-tL^2}{\frac{L}{K_2} t+1}=\frac{K_1K_2-tL^2}{\gamma t+1},
	\end{eqnarray*} where $t=\frac{K_1K_2M}{N} \in \{0,\frac{K_2}{L},\ldots,\lfloor\frac{K_1}{L}\rfloor \frac{K_2}{L}\}$ and $\gamma=\frac{L}{K_2}$.   In the 2D MACC system, each user can retrieve $L$ different subfiles which are stored in $L$ neighboring cache-nodes in each row. By the property of MDS code, each file could be recovered from any of its $L$ subfiles; thus the demand of each user is satisfied.
\end{proof}

Note that,  when $K_2>L$ 
 the coded caching gain of the scheme in Theorem \ref{th-baseline} is always less than $t$.  To improve this scheme,
 when $L|K_1$ and $L|K_2$, we can divide the cache-nodes into $L^2$ non-overlapping groups, each of which has $\frac{K_1K_2}{L^2}$ cache-nodes. By using the MN scheme for each group, we have the following scheme whose coded caching gain is $t+1$, whose proof could be found in Section~\ref{sec-proof of Theorem Group}.
\begin{theorem}\rm\emph{(Grouping Scheme)}
\label{th-Group}
For the $(K_1,K_2,L,M,N)$ MACC problem, when $L|K_1$ and $L|K_2$, the lower convex envelope of the following memory-load tradeoff corner points is achievable,
\begin{align}
(M, R_2)= \left(\frac{N t}{K_1K_2}, \frac{K_1K_2-tL^2}{t+1} \right), \ \forall t \in \left[0:\frac{K_1K_2}{L^2}\right].  \label{eq:load R1}
\end{align}
\hfill $\square$
\end{theorem}

Next, we will propose a highly non-trivial hybrid construction for the case $K_2>L$. 
This construction is consisted of outer and inner structures,  which is built on  a transformation approach from two classes of 1D MACC schemes (for 1D MACC systems in vertical and horizontal projections of the 2D MACC system, respectively) to a 2D MACC scheme.
As the outer structure (i.e., the vertical projection of 2D system), we could choose any $(K_1,L,M_1=K_2M,N)$ 1D MACC scheme  from  the transformation approach in Theorem~\ref{lem-1D-tran}.
As the inner structure  (i.e., the horizontal projection of 2D system),
  we choose $\frac{K_1K_2M}{N}$ different   $(K_2,L,M_2=N/K_2,N)$ 1D MACC schemes from the transformation approach   on Partition PDA in Corollary~\ref{lem-line-par}. 
\begin{theorem}\rm\emph{(Hybrid Scheme)}
\label{th-General-regular}
	For the $(K_1,K_2,L,M,N)$ MACC problem with $K_2>L$,
	given any $(K_1',F_1',Z_1',S_1')$ PDA (satisfying Conditions C$1$-C$5$) where $K_1'=K_1-t(L-1)$,
	and a $t$-$(tK_2,K_2^t,LK_2^t,K_2^t(K_2-L))$ Partition PDA,
	the lower convex envelope of the following memory-load tradeoff corner points is achievable,
	\begin{eqnarray}
	\begin{split}
	(M, R_3)=& \left(\frac{N t}{K_1K_2}, \frac{K_2tL-tL^2}{t}+\frac{K_2S_1'}{F_1'} \right), \ \forall t\in \left[ \left\lfloor \frac{K_1}{L} \right\rfloor \right],  \label{eq:grid-load-R3}
	\end{split}
	\end{eqnarray}
	and $(M,R_3)=(0, K_1K_2)$, $(M,R_3)=\left(\frac{N}{L^2},0\right)$.
	\hfill $\square$
\end{theorem}
The proof of  Theorem  \ref{th-General-regular} could be found  in Section \ref{sec-proof of Theorem 1}.

When the 1D MACC scheme  for the outer structure is generated from the MN PDA in Corollary~\ref{lem-line}, by applying the novel transformation approach in Theorem~\ref{th-General-regular}, the following result can be directly obtained.
\begin{theorem}\rm\emph{(Hybrid Scheme via MN PDA)}
\label{th-General}
For the $(K_1,K_2,L,M,N)$ MACC problem with $K_2>L$,
given a $(t+1)$-$\left(K_1',{K_1'\choose t},{K_1'-1\choose t-1},{K_1'\choose t+1}\right)$ MN PDA where $K_1'=K_1-t(L-1)$,
and a $t$-$(tK_2,K_2^t,LK_2^t,K_2^t(K_2-L))$ Partition PDA,
the lower convex envelope of the following memory-load tradeoff corner points is achievable,
\begin{eqnarray}
\begin{split}
(M, R_4)=& \left(\frac{N t}{K_1K_2}, \frac{K_2tL-tL^2}{t}+\frac{K_1K_2-K_2tL}{t+1} \right), \ \forall t\in \left[ \left\lfloor \frac{K_1}{L} \right\rfloor \right],  \label{eq:grid-load-R}
\end{split}
\end{eqnarray}
 and $(M,R_4)=(0, K_1K_2)$, $(M,R_4)=\left(\frac{N}{L^2},0\right)$.
\hfill $\square$
\end{theorem}

\begin{remark}\rm\emph{(Local caching gain and coded caching gain in Theorem~\ref{th-General})}
The local caching gain of the hybrid scheme in Theorem \ref{th-General}  is $1-\frac{L^2M}{N}$, which is the same as the local caching gains of the baseline scheme and the grouping scheme in Theorems~\ref{th-baseline} and \ref{th-Group}. In addition,   its coded caching gain is between $t$ (the denominator of the first item) and $t+1$ (the denominator of the last item).
\hfill $\square$
\end{remark}
We conclude this subsection with some numerical comparisons of the schemes in Theorems~\ref{th-baseline}, \ref{th-Group}, and~\ref{th-General}. In Fig~\ref{subfig-123}, we consider the case where $K_1=12$, $K_2=8$, $L=2$, and $N=96$. It can be seen that both schemes in Theorems~\ref{th-Group} and \ref{th-General} have lower loads than the baseline scheme in Theorem~\ref{th-baseline}. Furthermore, the load of the hybrid scheme in Theorem~\ref{th-General} is slightly larger than that of the grouping scheme in Theorem~\ref{th-Group}.
In Fig.~\ref{subfig-13}, we consider the case where $K_1=11$, $K_2=9$,   $L=2$, and $N=99$. Since $L$ does not divide $K_1$ nor $K_2$, the scheme in Theorem~\ref{th-Group} cannot be used.  It can be seen that the proposed hybrid scheme in Theorem~\ref{th-General} outperforms the baseline scheme in Theorem~\ref{th-baseline}.

\begin{figure}
\centering
\begin{subfigure}{0.45\textwidth}
\includegraphics[width=3.1in]{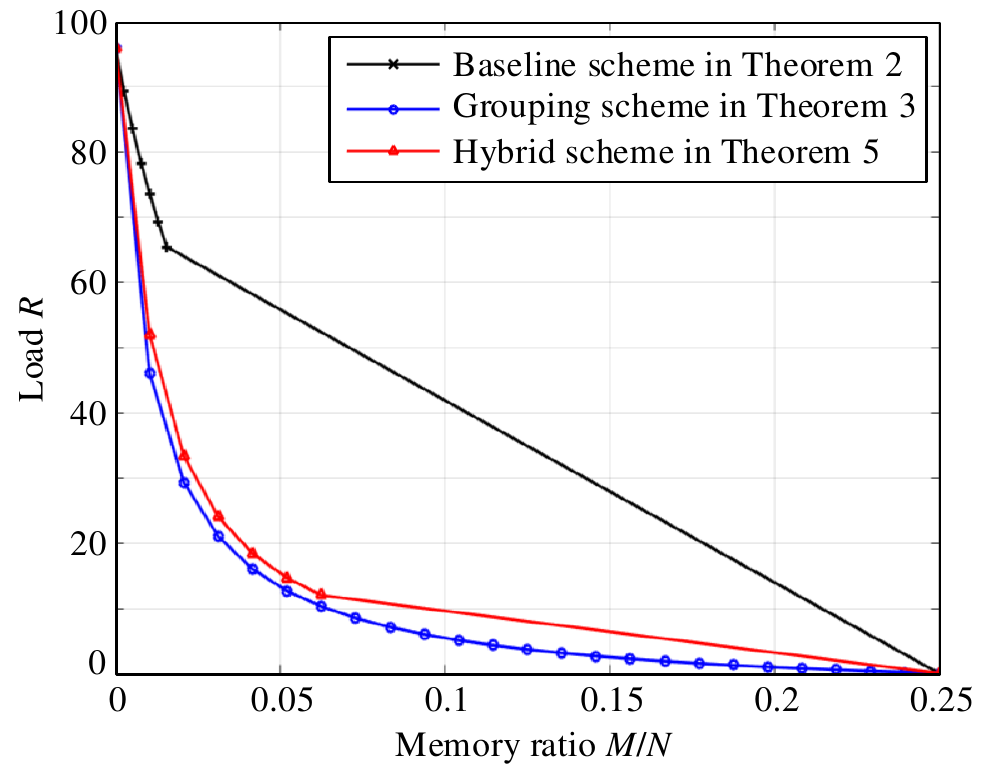}
\caption{\small  $K_1=12$, $K_2=8$, $L=2$ and $N=96$}
\label{subfig-123}
\end{subfigure} \ \ \ \ \ \ \
\begin{subfigure}{0.45\textwidth}
\includegraphics[width=3.1in]{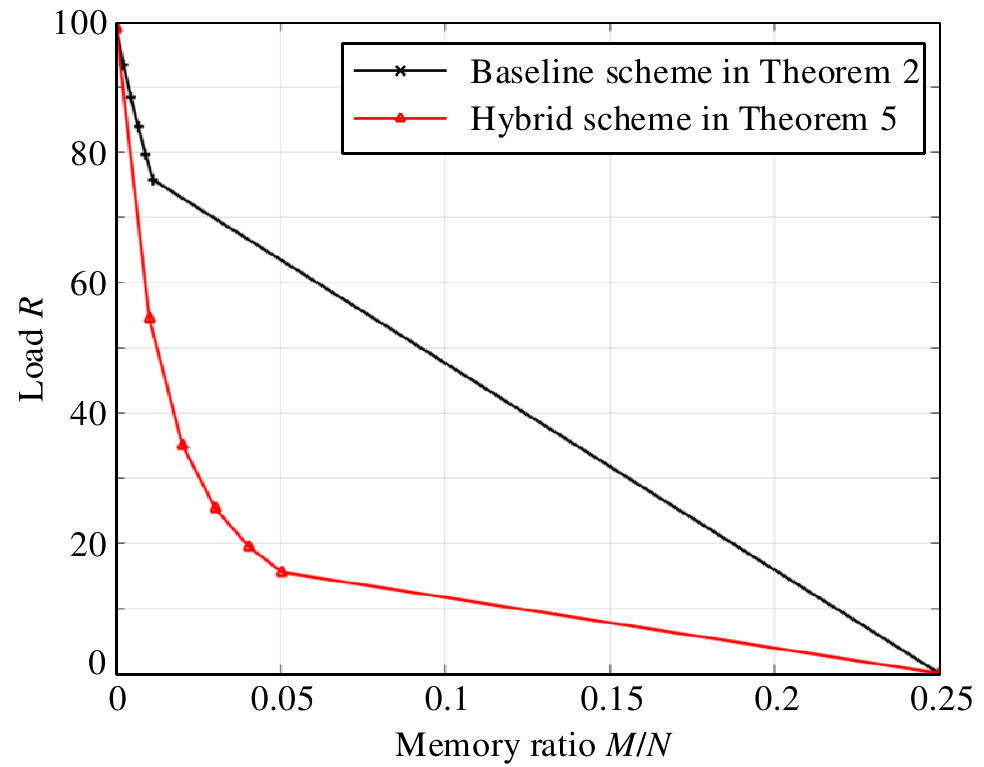}
\caption{\small  $K_1=11$, $K_2=9$, $L=2$ and $N=99$}
\label{subfig-13}
\end{subfigure}
\caption{The transmission load of the caching schemes in Theorems~\ref{th-baseline}, \ref{th-Group} and \ref{th-General}}\label{fig-numerical}
\end{figure}

\subsection{Example of the Grouping Scheme in Theorem~\ref{th-Group}}
\label{sub:illustrated example TH1}
Let us consider a $(K_1,K_2,L,M,N)=(4,4,2,1,16)$ 2D MACC problem. In this case, we have $L|K_1$ and $L|K_2$. 
\begin{itemize}
\item \textbf{Placement phase.} Each file is divided into $L^2=4$ subfiles with equal length, i.e., $W_{n}=\left\{W^{(1)}_{n},W^{(2)}_{n},W^{(3)}_{n},W^{(4)}_n \right\}$ where $n\in [16]$, and the cache-nodes are divided into $L^2=4$ groups, i.e.,
    \begin{eqnarray*}
    &&\mathcal{G}_1=\{C_{1,1},C_{1,3},C_{3,1},C_{3,3}\},\ \ \ \ \mathcal{G}_2=\{C_{1,2},C_{1,4},C_{3,2},C_{3,4}\},\\
    &&\mathcal{G}_3=\{C_{2,1},C_{2,3},C_{4,1},C_{4,3}\},\ \ \ \ \ \mathcal{G}_4=\{C_{2,2},C_{2,4},C_{4,2},C_{4,4}\},
    \end{eqnarray*} as illustrated in Fig~\ref{group}.
\begin{figure}
  \centering
  \includegraphics[width=2.5in]{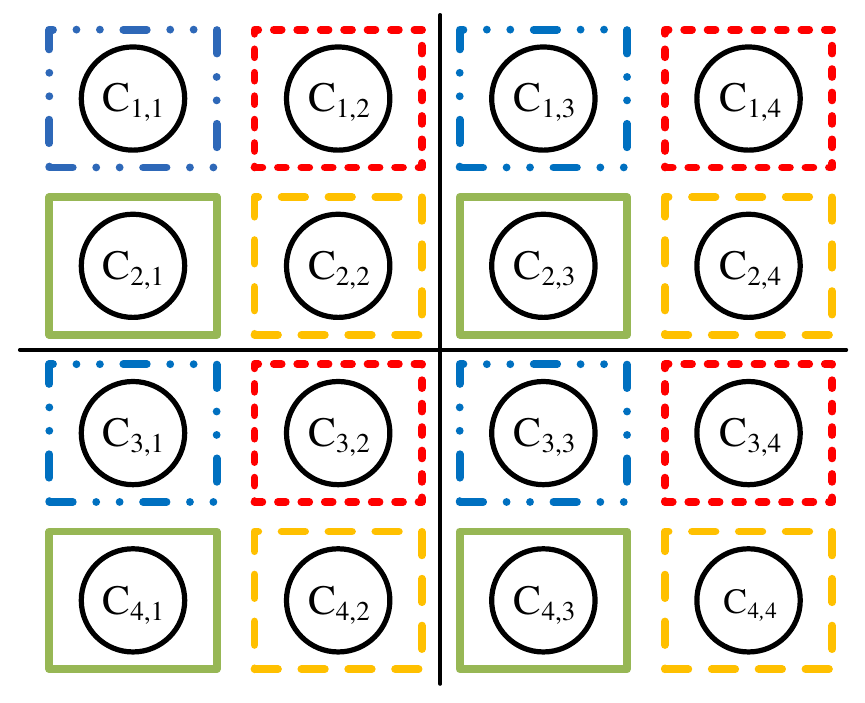}\\
  \caption{Groups of cache-nodes  in the grouping scheme.}\label{group}
\end{figure}
Then the server places the subfiles $\{W^{(l)}_{n}\ |\ n\in [16]\}$ to the cache-nodes in $\mathcal{G}_l$, $l\in [4]$, by  the placement phase of $(\frac{K_1K_2}{L^2},L^2M,N)=(4,4,16)$ MN scheme. Each cache-node caches $\frac{1}{L^2}\cdot{L^2M}=M=1$ file, satisfying the memory size constraint. Furthermore, any two cache-nodes connected to the common user do not cache the same content. Since each user can access $L^2$ cache-nodes, the local caching gain of the proposed scheme is $g_{\text{local}}=1-\frac{L^2M}{N}=\frac{3}{4}$. In each group, every user can retrieve one different subfile of each file.


\item \textbf{Delivery phase.} We divide the users into the following four groups according to the cache-node groups,
      \begin{eqnarray*}
    &&\mathcal{G}'_1=\{U_{1,1},U_{1,3},U_{3,1},U_{3,3}\},\ \ \ \ \mathcal{G}'_2=\{U_{1,2},U_{1,4},U_{3,2},U_{3,4}\},\\
    &&\mathcal{G}'_3=\{U_{2,1},U_{2,3},U_{4,1},U_{4,3}\},\ \ \ \ \ \mathcal{G}'_4=\{U_{2,2},U_{2,4},U_{4,2},U_{4,4}\}.
    \end{eqnarray*}
Let us focus on the $i^{\text{th}}$ group of users, where $i\in [4]$. The transmission for this group of users contains $4$ time slots.
In the $l^{\text{th}}$ time slot where $l\in [4]$, these users will use the cache content of the cache-nodes in $\mathcal{G}_l$.
The multicast messages in this time slot are generated through the  $(4,4,16)$ MN scheme on the   subfiles in $\{W^{(l)}_{n} \ | \ n\in [N]\}$ which are demanded by the users in $\mathcal{G}'_i$.

Since  the coded caching gain of the MN scheme is
    $  g_{\text{coded}}=\frac{K_1K_2M}{N}+1=2$,
    the load of the proposed scheme is
$  K_1K_2\frac{g_{\text{local}}}{g_{\text{coded}}}=6  $,
which coincides with \eqref{eq:load R1}.
\end{itemize}

\subsection{Example of Hybrid Scheme in Theorem~\ref{th-General}}
\label{sub:illustrated example TH3}
Let us consider the $(K_1,K_2,L,M,N)=(5,3,2,2,15)$ 2D MACC system.  The hybrid scheme in Theorem~\ref{th-General} consists
of an outer structure and an inner structure, which are generated from a scheme for the $(K_1,L,M_1=K_2M,N)=(5,2,6,15)$ 1D MACC problem (i.e., the 1D model in the vertical projection of the 2D model), and $\frac{K_1K_2M}{N}=2$ schemes for the $(K_2,L,M_2=N/K_2,N)=(3,2,5,15)$ 1D MACC problem (i.e., the 1D model in the horizontal projection of the 2D model), respectively. We choose these two classes of 1D MACC problems satisfying  $\frac{M_1}{N}\cdot\frac{M_2}{N}=\frac{M}{N}$.

We divide each file into $K_1=5$ equal-length subfiles, $W_n=\{W_n^{(r)}\ |\ r\in[5]\}$, and  divide the caching procedure into $5$ separate rounds. For each $r\in[5]$,  in the $r^{\text{th}}$  round we only consider the $r^{\text{th}}$ subfile of each file. Since all the caching procedures in different rounds are symmetric, we focus on the first round, and construct the node-placement array $\mathbf{C}$, user-retrieve array $\mathbf{U}$ and user-delivery array $\mathbf{Q}$, defined as follows.

\begin{definition}\rm
	\label{defn-three arrays}
	Given integers $F'$ and $K$ which represent the subpacketization of the first round and the number of cache-nodes (or users) respectively, we define that
	\begin{itemize}
		\item An $F' \times K$ node-placement array $\mathbf{C}$ consists of ``$*$" and null entries. The entry located at the position $(j,k)$ in $\mathbf{C}$ is star if and only if the $k^{\text{th}}$	  cache-node caches the $j^{\text{th}}$ packet of $W^{(1)}_n$ where $n\in [N]$. Note that,
		 the $K $ cache-nodes are ordered into $K$ columns of $\mathbf{C}$ 
		 as $(C_{1,1},C_{1,2},\ldots,C_{1,K_2},C_{2,1},\ldots,C_{K_1,K_2})$.
		\item An $F' \times K$ user-retrieve array $\mathbf{U}$ consists of ``$*$" and null entries. The entry located at the position $(j,k)$ in $\mathbf{U}$ is star if and only if the $k^{\text{th}}$   user can retrieve the $j^{\text{th}}$ packet of  $W^{(1)}_n$  where $n\in [N]$, from its connected cache-nodes. Note that,
		 the $K $ users are ordered into $K$ columns of $\mathbf{U}$ 
		 as $(U_{1,1},U_{1,2},\ldots,U_{1,K_2},U_{2,1},\ldots,U_{K_1,K_2})$.
		\item An $F' \times K$ user-delivery array $\mathbf{Q}$ consists of $\{*\}\bigcup [S]$, which is obtained by filling the null entries in  $\mathbf{U}$ by some integers.
		 Each integer represents a multicast message, while $S$ represents the total number of multicast messages transmitted in the first round during the delivery phase.
	\end{itemize}
	\hfill $\square$
\end{definition}

For the sake of clarity, we label the columns of $\mathbf{C}$, $\mathbf{U}$ and $\mathbf{Q}$ by vectors $(k_1,k_2)$ where $k_1\in[K_1]=[5]$ and $k_2\in[K_2]=[3]$.
 For each $k_1\in[K_1]=[5]$, we define  $(k_1,[3])$  as the column index set $\{(k_1,1), (k_1,2), (k_1,3) \}$.
 The constructions of $\mathbf{C}$, $\mathbf{U}$ and $\mathbf{Q}$ are listed as follows, as illustrated in Fig.~\ref{fig-construct-CUQ}.

\begin{figure*}
	\centering
	\vspace{-10pt}	
	\setlength{\abovecaptionskip}{0pt}
	$\begin{array}{c|c}
		\begin{subfigure}[t]{0.5\textwidth}
			\hspace{-1cm}
			\centering
			\setlength{\abovecaptionskip}{0pt}
			\setlength{\belowcaptionskip}{5mm}
			\includegraphics[width=3.2in]{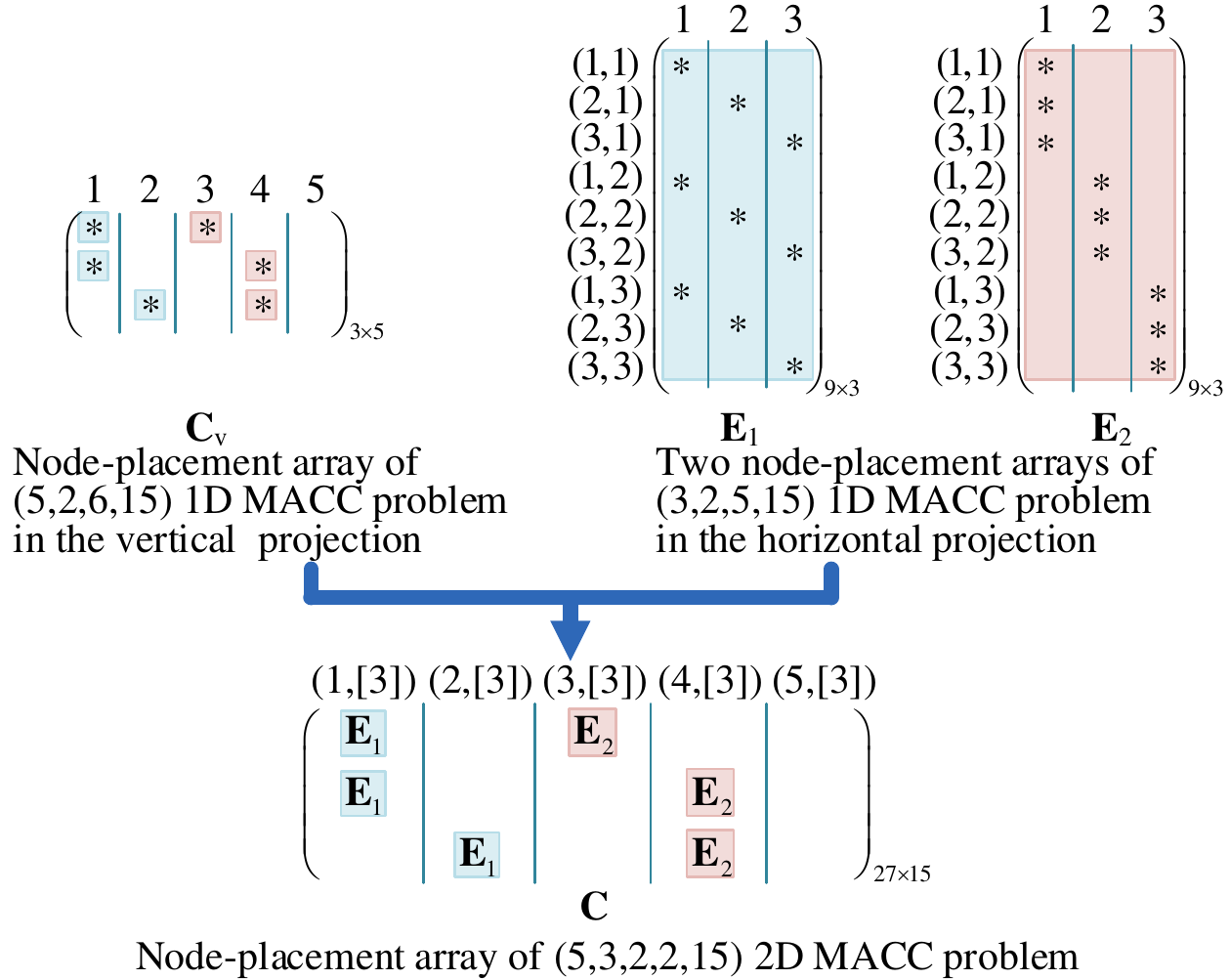}
			\caption{Node-placement array $\mathbf{C}$}
			\label{fig-C}
		\end{subfigure} &
		\begin{subfigure}[t]{0.5\textwidth}
			\centering
			\setlength{\abovecaptionskip}{0pt}
			\setlength{\belowcaptionskip}{5mm}
			\includegraphics[width=3.2in]{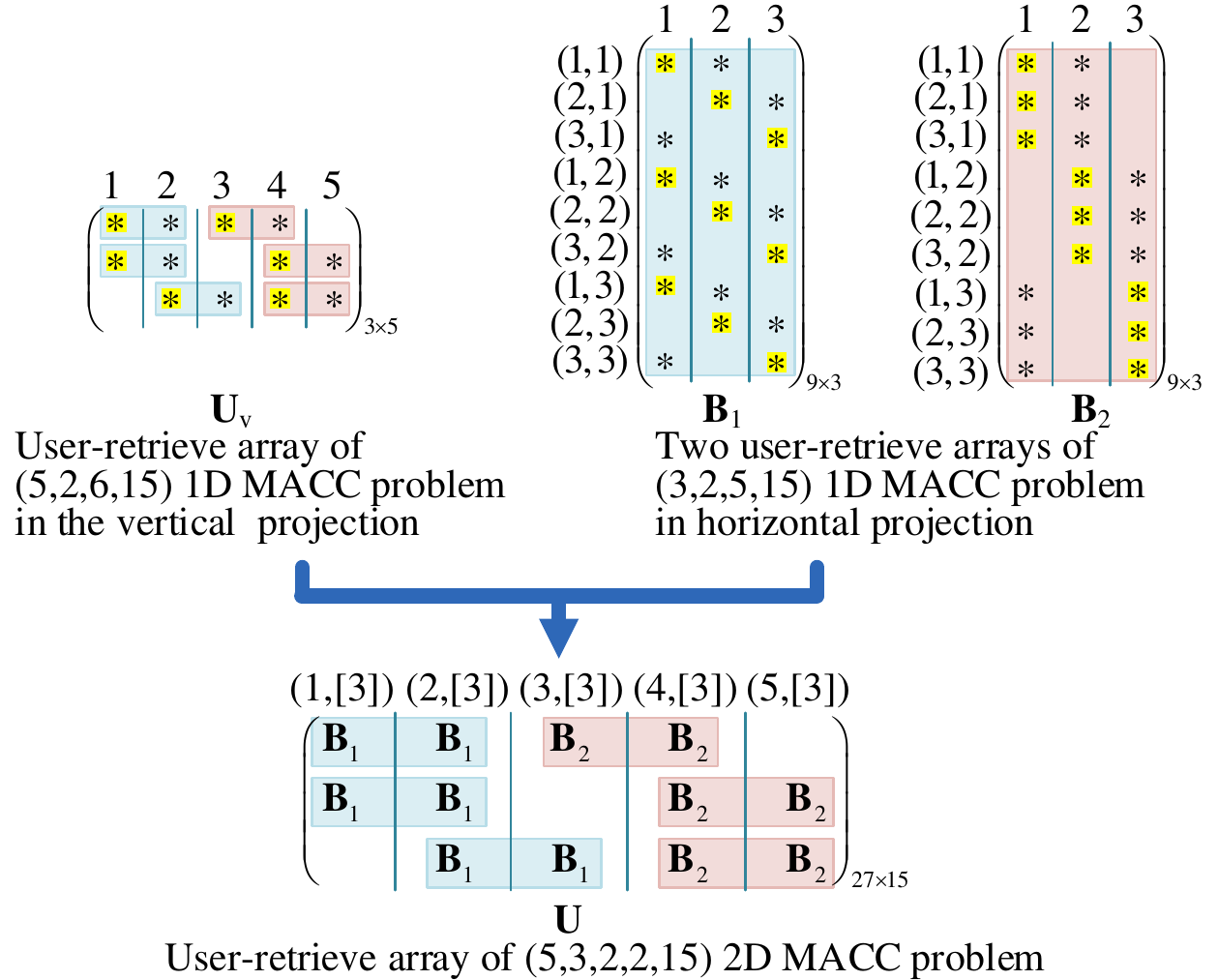}
			\caption{User-retrieve array $\mathbf{U}$}
			\label{fig-U}
		\end{subfigure}  	
     \end{array}$
	
	\begin{subfigure}[t]{1\textwidth}
		\vspace{0.5cm}
		\hspace{-1.7cm}
		\centering
		\setlength{\abovecaptionskip}{0pt}
		\includegraphics[width=5.5in]{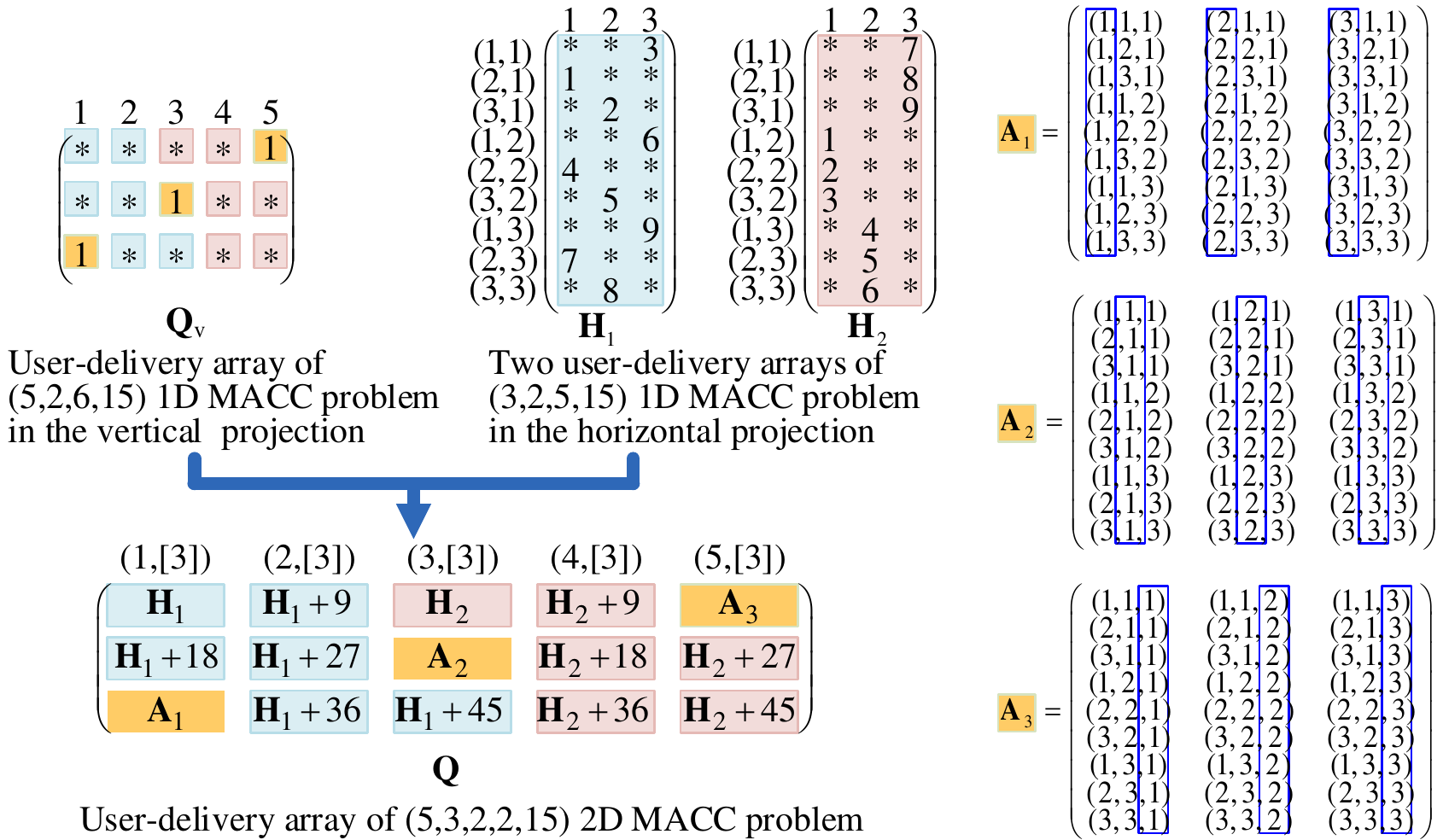}
		\caption{User-delivery array $\mathbf{Q}$}
		\label{fig-Q}
	\end{subfigure}
	\vspace{3mm}
	\caption{Flow diagram of constructing $\mathbf{C}$, $\mathbf{U}$  and $\mathbf{Q}$ for $(5,3,2,2,15)$ 2D MACC system where $(\mathbf{H}_1,\mathbf{H}_2)$ is a Partition PDA}\label{fig-construct-CUQ}
	\vspace{-8mm}
\end{figure*}

\begin{itemize}
	\item \textbf{The construction of node-placement array $\mathbf{C}$}.  As illustrated in Fig.~\ref{fig-C}, the node-placement array $\mathbf{C}$ is designed via
	outer  and inner structures, respectively.  $\mathbf{C}$ is composed of an outer structure which corresponds to a vertical 1D MACC problem containing $K_1=5$ cache-nodes. We select the node-placement array of the $(K_1,L,M_1,N)=(5,2,6,15)$  1D MACC scheme $\mathbf{C}_{\text{v}}$ (detailed in Fig.~\ref{fig-line-macc} ) as the outer structure.
	 We then extend $\mathbf{C}_{\text{v}}$ into the 2D  MACC node-placement array $\mathbf{C}$ by replacing each entry in $\mathbf{C}_{\text{v}}$ by an inner node-placement array with $K_2=3$ columns. More precisely:
	 \begin{itemize}
	 \item For the stars in each row of $\mathbf{C}_{\text{v}}$, we  replace the first star  by $\mathbf{E}_1$, and replace the last star by $\mathbf{E}_2$. Note that, the inner structure $\mathbf{E}_1$ and $\mathbf{E}_2$ are node-placement arrays (detailed in Fig.~\ref{fig-par-macc}) of the $(K_2,L,M_2,N)=(3,2,5,15)$ 1D MACC problem in the horizontal projection of the 2D model.
	 \item For the null entry in each row of $\mathbf{C}_{\text{v}}$, we replace it by a null array with dimension $9\times 3$ which is the same as $\mathbf{E}_1$ and $\mathbf{E}_2$.
	 \end{itemize}
	{\bf By this construction of $\mathbf{C}$,  any two cache-nodes connected to common users do not cache any common packets.} In other words,
	 for each row of $\mathbf{C}$ (representing each packet),
 any two stars at column $(k_1,k_2)$ and column $(k_1',k_2')$ satisfying $D_{\rm r}(k_1-k_1') \geq 2$.
 \footnote{\label{foot-D-distance}We define that $D_{\rm r}(k_1,k_1')=\min\{<k_1-k_1'>_{K_1},K_1-<k_1-k_1'>_{K_1}\}$.}
		
	\item \textbf{The construction of user-retrieve array $\mathbf{U}$}. Once the node-placement array $\mathbf{C}$ is designed, the user-retrieve array $\mathbf{U}$ is also determined, as illustrated in Fig.~\ref{fig-U}.
	In the 2D MACC system, each user can access $L^2$ cache-nodes satisfying~\eqref{eq-distance-grid}. In other words, focus on the same row of $\mathbf{C}$ and $\mathbf{U}$,  if the column $(k_1',k_2')$ in $\mathbf{C}$ is ``*", then the column $(k_1,k_2)$ in $\mathbf{U}$ is set to be ``*" where $  <k_1-k_1'>_{5} \ <2$ and $<k_2-k_2'>_{3} \ <2$.	 	
	From the design of $\mathbf{C}$, the outer structure of  $\mathbf{U}$   corresponds to   the user-retrieve array $\mathbf{U}_{\text{v}}$ of the $(5,2,6,15)$  1D MACC scheme detailed in Fig.~\ref{fig-line-macc}. 
 We then extend $\mathbf{U}_{\text{v}}$ into the 2D MACC user-retrieve array $\mathbf{U}$ by replacing each entry in $\mathbf{U}_{\text{v}}$ by an inner user-retrieve array with $K_2=3$ columns. More precisely:
	\begin{itemize}
		\item For the stars in each row of $\mathbf{U}_{\text{v}}$, we replace the first $L=2$ consecutive stars by $\mathbf{B}_1$, and replace the last $L=2$ consecutive stars by $\mathbf{B}_2$. Note that $\mathbf{B}_1$ and $\mathbf{B}_2$ are the user-delivery arrays for the inner structure detailed in Fig.~\ref{fig-par-macc},  which are  from the $(K_2,L,M_2,N)=(3,2,5,15)$ 1D MACC problem in the horizontal projection of the 2D model. Since there are $3$ rows of $\mathbf{U}_{\text{v}}$, we obtain $3\times L=6$ $\mathbf{B}_1$ and $6$ $\mathbf{B}_2$ of $\mathbf{U}$, respectively.
		\item For the null entry in each row of $\mathbf{U}_{\text{v}}$, we replace it by a null array with dimension $9\times 3$.
	\end{itemize}		

	\item \textbf{The construction of user-delivery array $\mathbf{Q}$}. $\mathbf{Q}$ is obtained by  filling the null entries of $\mathbf{U}$ such that the Condition C3 of PDA in Definition~\ref{def-PDA} is satisfied.
	As illustrated in Fig.~\ref{fig-Q}, we design $\mathbf{Q}$ from $\mathbf{U}$ in two steps:
	\begin{itemize}
	\item In the first step, we fill the null entries in the inner structure of $\mathbf{U}$, i.e., the null entries of $\mathbf{B}_1$ and $\mathbf{B}_2$. 
	 Recall that $\mathbf{H}_1$ and $\mathbf{H}_2$ are user-delivery arrays of the $(3,2,5,15)$ 1D MACC problem detailed in Fig.~\ref{fig-par-macc}, which correspond to sub-Partition PDAs of $\mathbf{H}=(\mathbf{H}_1,\mathbf{H}_2)$ in Fig.~\ref{fig-par-inte}.
 We replace the 6 arrays $\mathbf{B}_1$ and 6 arrays $\mathbf{B}_2$ in $\mathbf{U}$ by $\mathbf{H}_1+9v$ and $\mathbf{H}_2+9v$ for each $v\in[0:5]$ from left to right then from top to bottom, respectively.
For example, the first $\mathbf{B}_1$ is replaced by $\mathbf{H}_1$ and the second $\mathbf{B}_1$ is replaced by $\mathbf{H}_1+9$, because there are $9$ different integers in $\mathbf{H}_1$.\footnote{\label{foot:H+S}  Recall that for any integer $a$, $\mathbf{H}_1+a$ denotes an array $(\mathbf{H}_1(j,k)+a)$, where $*+a=*$.}
Since  $(\mathbf{H}_1+a, \mathbf{H}_2+a)$ constitutes a Partition PDA, this integer-filling is valid for the conditions of PDA.
As a result, we have used $9\times6=54$  integers, each of which   occurs $t=2$ times.
{\bf Hence,  the coded caching gain for the multicast messages	in the first step (referred to as Type I multicast messages) is $t=2$.}

	\item In the second step, we fill the null entries in the outer structure of $\mathbf{U}$ 
	 by   arrays $\mathbf{A}_1$, $\mathbf{A}_2$, and $\mathbf{A}_3$, each of which has    dimension $9\times 3$.
For each $j\in[3]$, we label  each of 	 $27$ entries in $\mathbf{A}_j$ by a $t+1=3$-dimensional vector $\mathbf{e}=(e_1,e_2,e_3)$,
where $e_1,e_2,e_3\in[K_2]=[3]$. 
In $\mathbf{Q}$, as illustrated in Fig.~\ref{fig-Q},
the sets of column indices of $\mathbf{A}_1$, $\mathbf{A}_2$, and $\mathbf{A}_3$  are $(1,[3])$, $(3,[3])$ and $(5,[3])$\footnote{Recall that for each $k_1\in[K_1]=[5]$, we define $(k_1,[3])$ as the column index set $\{(k_1,1),(k_1,2),(k_1,3)\}$.}, respectively.
In $\mathbf{A}_1$, any vector $\mathbf{e}=(e_1,e_2,e_3)$ is filled in the entry indexed by row $(e_2,e_3)$ and column $e_1$ of $\mathbf{A}_1$.
In $\mathbf{A}_2$, any vector $\mathbf{e}=(e_1,e_2,e_3)$ is filled in the entry indexed by row $(e_1,e_3)$ and column $e_2$ of $\mathbf{A}_2$.
In $\mathbf{A}_3$, any vector $\mathbf{e}=(e_1,e_2,e_3)$ is filled in the entry indexed by row $(e_1,e_2)$ and column $e_3$ of $\mathbf{A}_3$.
By the above construction, it can be checked  that Condition C3 of PDA in Definition~\ref{def-PDA} is satisfied.
For instance,  let us focus on the case $\mathbf{e}=(3,2,1)$ filled in  $\mathbf{A}_1$, $\mathbf{A}_2$, and $\mathbf{A}_3$. The sub-array of  $\mathbf{Q}$ containing  the vector $(3,2,1)$ is denoted by $\mathbf{Q}_{(3,2,1)}$. We will show that $\mathbf{Q}_{(3,2,1)}$ is with the form
  illustrated in Fig.~\ref{fig-Q-sub},  and thus satisfies Condition C3 of PDA in Definition~\ref{def-PDA}.
In $\mathbf{A}_1$ the vector $(3,2,1)$ is filled in the entry indexed by row $(2,1)$ and column $3$;
 in $\mathbf{A}_2$ the vector $(3,2,1)$ is filled in the entry indexed by row $(3,1)$ and column $2$; in $\mathbf{A}_3$ the vector $(3,2,1)$ is filled in the entry indexed by row $(3,2)$ and column $1$.
By Definition~\ref{rem-par} of the Partition PDA, the entry at row $(2,1)$ and column $2$ of $\mathbf{H}_1$ is tag-star; the entry at row $(2,1)$ and column $1$ of $\mathbf{H}_2$ is tag-star. Thus we obtain  row $(2,1)$  of $\mathbf{Q}_{(3,2,1)}$.
Similarly,   the entry at row $(3,1)$ and column $3$ of $\mathbf{H}_1$ is tag-star; the entry at row $(3,1)$ and column $1$ of $\mathbf{H}_2$ is tag-star. Thus we obtain row $(3,1)$ of $\mathbf{Q}_{(3,2,1)}$.
The entry at row $(3,2)$ and column $3$ of $\mathbf{H}_1$ is tag-star; the entry at row $(3,2)$ and column $2$ of $\mathbf{H}_2$ is tag-star. Thus we obtain row $(3,2)$ of $\mathbf{Q}_{(3,2,1)}$. So the multicast message for the vector  $(3,2,1)$ is decodable.
\begin{figure}
	\centering
	\includegraphics[width=2in]{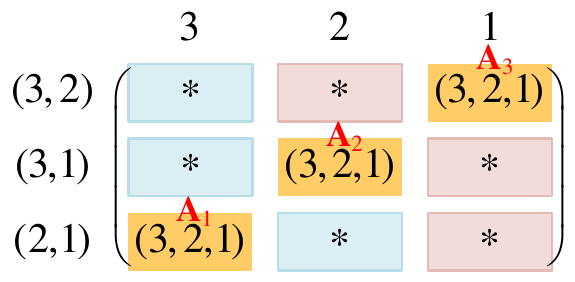}\\
	\caption{The sub-array $\mathbf{Q}_{(3,2,1)}$  containing the vector $(3,2,1)$ in $\mathbf{A}_1$, $\mathbf{A}_2$, $\mathbf{A}_3$.}\label{fig-Q-sub}
\end{figure}
As a result, we have used $27$ vectors, each of which occurs $t+1=3$ times. \bf{Hence, the coded caching gain for the multicast messages in the second step (referred to as Type II multicast messages) is $t+1=3$.}
\end{itemize}
\end{itemize}

After determining $\mathbf{C}$, $\mathbf{U}$, and $\mathbf{Q}$, the placement and delivery strategies are obtained during the first round. For each $r\in[K_1]=[5]$,  in the $r^{\text{th}}$ round,  we only need to right-shift $\mathbf{C}$, $\mathbf{U}$, and $\mathbf{Q}$ by $K_2(r-1)=3(r-1)$ positions in a cyclic wrap-around fashion.
Denoting  the node-placement array for the $r^{\text{th}}$  round   by $\mathbf{C}^{(r)}$,  the overall placement array of the cache-nodes is $[\mathbf{C}^{(1)};\mathbf{C}^{(2)};\ldots;\mathbf{C}^{(5)}]$. Since each $\mathbf{C}_{\text{v}}$ has $6$ stars
, and each column of $\mathbf{E}_1$ and $\mathbf{E}_2$ has $3$ stars, the number of stars in each column of $[\mathbf{C}^{(1)};\ldots;\mathbf{C}^{(5)}]$ is $6\times3=18$. In addition, $[\mathbf{C}^{(1)};\ldots;\mathbf{C}^{(5)}]$ has $3\times 9\times 5=135$ rows. Thus each cache-node caches $M=\frac{18}{135}\times N=\frac{2}{15}N=2$ files, satisfying the memory size constraint.

According to the placement strategy, in the hybrid scheme, any two cache-nodes connected to some common users do not cache the same packets. Since each user can access $L^2=4$ cache-nodes in the 2D MACC system, the local caching gain of the proposed scheme is
     $g_{\text{local}}= 1-\frac{L^2 M}{N}=\frac{7}{15}.$
In the delivery phase,  there are $K_1\times54=5\times54=270$ multicast messages in Type I with the coded caching gain
$ g_{\text{I}}=t=2$;
  there are $K_1\times27=5\times27=135$ multicast messages in Type II with the coded caching gain
     $g_{\text{II}}=t+1=3$.
So the overall coded gain of the hybrid scheme is
  $g_{\text{coded}}= \frac{270\times 2+135\times 3}{270+135}=\frac{7}{3}$; thus the achieved
  load is
 $ K_1 K_2 \frac{g_{\text{local}}}{g_{\text{coded}}} =15\times \frac{7 / 15}{7 / 3}=3$,
which coincides with~\eqref{eq:grid-load-R}.

\section{Proof of Theorem \ref{th-Group}}
\label{sec-proof of Theorem Group}
In this section, we describe the grouping scheme for $(K_1,K_2,L,M,N)$ 2D MACC system under the constraints $L|K_1$ and $L|K_2$. Cache-nodes and users are divided into the following $L^2$ groups respectively,
\begin{eqnarray*}
	&&\mathcal{G}_{j_1,j_2}=\left\{C_{k_1,k_2} \ |\ k_1=j_1+i_1L, k_2=j_2+i_2L, i_1\in\left[0:\frac{K_1}{L}\right), i_2\in\left[0:\frac{K_2}{L}\right)\right\},\\
	&&\mathcal{G}'_{j_1',j_2'}=\left\{U_{k_1',k_2'}\ |\ k_1'=j_1'+i_1'L, k_2'=j_2'+i_2'L, i_1'\in\left[0:\frac{K_1}{L}\right), i_2'\in\left[0:\frac{K_2}{L}\right)\right\},
\end{eqnarray*}
where $j_1$, $j_1'$, $j_2$, $j_2'\in [L]$. Then $|\mathcal{G}_{j_1,j_2}|=|\mathcal{G}'_{j_1',j_2'}|=\frac{K_1K_2}{L^2}:=\widehat{K}$.

\begin{itemize}
	\item \textbf{Placement phase}. Each file is divided into $L^2$ subfiles with equal length, i.e., $W_{n}=\big\{ W^{(j_1,j_2)}_{n}\ | \ j_1,j_2\in [L]\big\}$. Define that $\mathcal{W}^{(j_1,j_2)}=\big\{W^{(j_1,j_2)}_{1}, W^{(j_1,j_2)}_{2}, \ldots, W^{(j_1,j_2)}_{N}\big\}$ for each $j_1,j_2\in [L]$. The server places the  subfiles in $\mathcal{W}^{(j_1,j_2)}$ to the cache-nodes in $\mathcal{G}_{j_1,j_2}$,  by the placement phase of the $(\widehat{K},\widehat{M},N)$ MN scheme where $\widehat{M}=L^2M$. Each cache-node caches $\widehat{M}\cdot\frac{1}{L^2}=M$ files, satisfying the memory size constraint. Furthermore, any two cache-nodes connected to the common user (i.e., satisfying \eqref{eq-distance-grid}) do not cache the same content.
	\item \textbf{Delivery phase}. Focus on the users in  group $\mathcal{G}'_{j_1',j_2'}$, where $j_1'$, $j_2'\in [L]$. The transmission for this group of users contains $L^2$ time slots. Each time slot is indexed by $(j_1,j_2)$, where $j_1$, $j_2\in [L]$. In the time slot $(j_1,j_2)$, the users will use the cache content stored by the cache-nodes in $\mathcal{G}_{j_1,j_2}$. The multicast messages in this time slot are generated through the $(\widehat{K},\widehat{M},N)$ MN scheme on the subfiles in $\mathcal{W}^{(j_1,j_2)}$ which are demanded by the users in $\mathcal{G}'_{j_1',j_2'}$. Since there are $L^2$ groups of users, the transmission load is
	$$R_2=L^2\frac{\widehat{K}(1-\widehat{M}/N)}{\widehat{K}\widehat{M}/N+1}=\frac{K_1K_2-tL^2}{t+1},$$
	where $t=\frac{K_1K_2M}{N}\in[0:\frac{K_1K_2}{L^2}]$.
	\item \textbf{Decodability}. In the 2D MACC system, each user can access all the $L^2$ cache-nodes satisfying \eqref{eq-distance-grid} which cache $L^2$ different subfiles. Hence, each user can totally obtain $L^2$ subfiles of each file from the placement and delivery phases, such that it can decode its desired file.
\end{itemize}

Hence, we proved the Theorem~\ref{th-Group}.

\section{Proof of Theorem  \ref{th-General-regular}}
\label{sec-proof of Theorem 1}
In this section, we describe the hybrid scheme (i.e., consisting of outer and inner structures) for the $(K_1,K_2,L,M,N)$ 2D MACC system,  where $t=\frac{K_1 K_2M}{N} \in \left[\left\lfloor \frac{K_1}{L} \right\rfloor \right]$.
The hybrid scheme constructs a novel transformation approach to generate 2D MACC scheme from two classes 1D MACC schemes, where the outer structure corresponds to a scheme for the $(K_1,L,M_1=K_2M,N)$ 1D MACC problem in vertical projection is generated by the transformation approach  proposed in \cite{CWLZC} and detailed in Section~\ref{subsect-line}; and the inner structure corresponds to $t$ different schemes for $(K_2,L,M_2=N/K_2,N)$ 1D MACC problem in the horizontal projection are generated by using the transformation approach \cite{CWLZC} on the  Partition PDA in \cite{YCTC} for the shared-link caching model.

In the hybrid scheme, we divide each file $W_{n}$ where $n\in [N]$ into $K_1$ subfiles with equal length, i.e., $W_{n}=\big\{W^{(1)}_{n},\ldots,W^{(K_1)}_n \big\}$. Denote the set of the $r^{\text{th}}$ subfiles by $\mathcal{W}^{(r)}=\big\{W^{(r)}_{1}, \ldots, W^{(r)}_{N}\big\}$ for each $r\in [K_1]$.
  We divide the whole caching procedure into $K_1$ separate rounds, where in the $r^{\text{th}}$ round we only deal with $\mathcal{W}^{(r)}$.
  Our construction contains three steps: the generations of  node-placement array $\mathbf{C}$,  user-retrieve array $\mathbf{U}$, and  user-delivery array $\mathbf{Q}$, respectively.

\subsection{Caching Strategy for Cache-nodes: Generation of the Array  $\mathbf{C}$}
\label{subsec-caching-nodes}
From Section~\ref{sub:illustrated example TH3}, we first construct the node-placement array $\mathbf{C}$ of the $(K_1,K_2,L,M,N)$ 2D MACC system in the first round. $\mathbf{C}$ is designed via outer and inner structures, respectively. We select the node placement array $\mathbf{C}_{\text{v}}$ (defined in \eqref{eq-array-node-caching}) for the outer structure which corresponds to the $(K_1,L,M_1=K_2M,N)$ 1D MACC problem in the vertical projection of the 2D model, and then extend $\mathbf{C}_{\text{v}}$ into the 2D MACC node-placement array $\mathbf{C}$ by replacing each entry in $\mathbf{C}_{\text{v}}$ by an inner node-placement array with $K_2$ columns. More precisely, for the $t$ stars in each row of $\mathbf{C}_{\text{v}}$, we replace them (from left to right) by  $\mathbf{E}_1,\mathbf{E}_2,\ldots,\mathbf{E}_{t}$, each of which corresponds to a node-placement array of the $(K_2,L,M_2=N/K_2,N)$ 1D MACC problem  in the horizontal projection of the 2D model. For null entries in $\mathbf{C}_{\text{v}}$, we replace each of them by a null array with dimension $F_2\times K_2$ (i.e., with the same dimension as $\mathbf{E}_1,\ldots,\mathbf{E}_{t}$), where $F_2=K_2^t$ equals to the subpacketization of Partition PDA that detailed in Section~\ref{sub:ori model}.
Then we get the node-placement array $\mathbf{C}$ in the first round.
By this construction, any two cache-nodes connected to some common users do not cache the same packet, since any two stars at column $(k_1,k_2)$ and column $(k_1',k_2')$ satisfying $D_{\rm r}(k_1,k_1') \geq L$ from \eqref{eq-caching-index}.

For each  $r\in [K_1]$, in the $r^{\text{th}}$ round, the node-placement array $\mathbf{C}^{(r)}$ is generated by   cyclically right-shifting $\mathbf{C}$ by $(r-1)K_2$ positions. In our hybrid construction, we use $\big[\mathbf{C}^{(1)};\mathbf{C}^{(2)};\ldots;\mathbf{C}^{(K_1)}\big]$ to represent the overall placement array of the cache-nodes. Since each outer structure  has $Z_1'K_1'$ stars in each column, and each inner structure   has $Z_2$ stars in each column, the number of stars in each column of $\big[\mathbf{C}^{(1)};\ldots;\mathbf{C}^{(K_1)}\big]$ is  $Z_1'K_1'Z_2$. In addition,  $\big[\mathbf{C}^{(1)};\ldots;\mathbf{C}^{(K_1)}\big]$ has $F_1' F_2 K_1$ rows. Thus, the   memory size of each cache-node is
\begin{align*}
\frac{Z_1'K_1'Z_2}{F_1' F_2 K_1}N=\frac{K_1'Z_1'}{F_1'}\cdot\frac{K_2Z_2}{F_2}\cdot\frac{1}{K_1K_2}N=\frac{K_1M_1}{N}\cdot1\cdot\frac{1}{K_1K_2}N =M,
\end{align*} satisfying the memory size constraint.

Note that when $(K_1,K_2,L,M,N)=(5,3,2,2,15)$, the above construction on $\mathbf{C}$ is illustrated in Fig.~\ref{fig-C}.
\subsection{Packets Retrievable to Users: Generation of the Array $\mathbf{U}$}
\label{subsect-packet-user}
After constructing the node-placement array $\mathbf{C}$, the user-retrieve array $\mathbf{U}$ is determined, since each user can access $L^2$ cache-nodes satisfying that the row and column modular distances are less than $L$. In other words, for the same row of $\mathbf{C}$ and $\mathbf{U}$, if the column $(k_1',k_2')$ of $\mathbf{C}$ is ``*", then the column $(k_1,k_2)$ of $\mathbf{U}$ is set to be ``*" where $<k_1-k_1'>_{K_1}<L$ and $<k_2-k_2'>_{K_2}<L$.

From the design of $\mathbf{C}$, we select the user-retrieve array $\mathbf{U}_{\text{v}}$ (defined in \eqref{eq-array-user-caching}) for the outer structure which corresponds to the $(K_1,L,M_1,N)$ 1D MACC problem in the vertical projection, then extend $\mathbf{U}_{\text{v}}$ to the 2D MACC user-retrieve array $\mathbf{U}$ by replacing each entry in $\mathbf{U}_{\text{v}}$ by an inner node-placement array with $K_2$ columns. More precisely, focus on each row of $\mathbf{U}_{\text{v}}$:
\begin{itemize}
\item There are $tL$ stars in this row; as shown in Section~\ref{subsect-line}, we can divide  these stars  into $t$ disjoint groups, each of which has $L$ consecutive stars.
For each $i\in[t]$, we replace each of the stars in the $i^{\text{th}}$ group (defined in~\eqref{eq:star group}) by  $\mathbf{B}_i$.\footnote{\label{foot:recall Bi} Recall that   $\mathbf{B}_1,\ldots,\mathbf{B}_t$  correspond to $t$ different user-retrieve arrays  of the $(K_2,L,M_2,N)$ 1D MACC problem in the horizontal projection.}
\item For null entries in this row, we replace each of which by a null array with dimension $F_2\times K_2$.
\end{itemize}

Then we get the user-retrieve array $\mathbf{U}$ in the first round.
Since each outer structure   has $tL$ stars in each row and
each inner structure   has $L$ stars in each row,
 the number of stars in each row of $\mathbf{U}$ is $tL^2$. So, there are $K_1K_2-tL^2$ null entries in each row of $\mathbf{U}$.
\begin{remark}\rm
\label{remark-2}
The null entries in each row of $\mathbf{U}$ can be divided into two disjoint parts.
\begin{itemize}
\item {\em Type I} : The inner structure null entries, i.e., the null entries  in
 $\mathbf{B}_{i}$ for all $i\in [t]$  of $\mathbf{U}$.  
Since  there are $K_2-L$ null entries in each row of $\mathbf{B}_{i}$, and  $L$ stars in each row of $\mathbf{U}_{\text{v}}$ are replaced by $\mathbf{B}_{i}$, thus there are totally $t(K_2-L)L$   null entries in each row of $\mathbf{U}$ in Type I.
\item {\em Type II}: The outer structure null entries, i.e., the null entries not in any $\mathbf{B}_{i}$ for all  $i\in [t]$  of $\mathbf{U}$. Since there are $K_1-tL$ null entries in each row of $\mathbf{U}_{\text{v}}$, and each of which is replaced by a null array with $K_2$ columns in $\mathbf{U}$, thus there are $K_2(K_1-tL)$ null entries in each row of $\mathbf{U}$ in Type II.
\end{itemize}
\hfill $\square$
\end{remark}
For each $r\in [K_1]$, in the $r^{\text{th}}$  round, the user-retrieve array $\mathbf{U}^{(r)}$ is obtained by cyclically right-shifting $\mathbf{U}$ by $(r-1)K_2$ positions.

Note that   when $(K_1,K_2,L,M,N)=(5,3,2,2,15)$, the above construction on $\mathbf{U}$ is illustrated in Fig.~\ref{fig-U}.

\subsection{Delivery Strategy: Generation of the Array $\mathbf{Q}$}
\label{subsect-delivery-user}
The user-delivery array $\mathbf{Q}$ is obtained by filling the null entries of $\mathbf{U}$ such that   Condition C3 of PDA in Definition~\ref{def-PDA} is satisfied. Inspired from Remark \ref{remark-2}, we fill the null entries in two steps:
\subsubsection{Step 1. Fill the null entries in Type I}\label{subsubsect-Type I}
\
\newline \indent
From Remark \ref{remark-2}, the entries in Type I are exactly the null entries of $\mathbf{B}_{i}$ where $i\in[t]$. Recall that for each $i\in[t]$, the user-retrieve array $\mathbf{B}_{i}$ has the same star entries as $\mathbf{H}_{i}$ of Partition PDA $\mathbf{H}$. 
We fill the null entries in Type I by replacing all $F_1'L$ arrays $\mathbf{B}_{i}$ in $\mathbf{U}$\footnote{\label{foot:F1'} Recall that $F_1'$ represents the number of rows in $\mathbf{U}_{\text{v}}$. In each row and for each $i\in[t]$, there are exactly $L$ stars replaced by $\mathbf{B}_{i}$ to obtain  $\mathbf{U}$. Thus for each $i\in[t]$, there are $F_1'L$ arrays $\mathbf{B}_i$ in $\mathbf{U}$.} by $\mathbf{H}_{i}+vK_2^t(K_2-L)$\footnote{\label{foot-H+S2}  Recall that for any integer $a$, $\mathbf{H}_1+a$ denotes an array $(\mathbf{H}_1(j,k)+a)$, where $*+a=*$.} (from left to right, from top to bottom) for each $i\in[t]$, where $v\in[0:F_1'L-1]$ and $K_2^t(K_2-L)$ is the number of different integers in $\mathbf{H}_{i}$.



For each  $v\in[0:F_1'L-1]$, since $(\mathbf{H}_1+vK_2^t(K_2-L),   \ldots, \mathbf{H}_t+vK_2^t(K_2-L))$ constitutes a Partition PDA, this integer-filling scheme  satisfies  Condition C3 of Definition \ref{def-PDA}.

\begin{remark}\rm
\label{remark-Q_v-Q}
From Fig.~\ref{fig-Q} and Section~\ref{subsect-packet-user}, the filling rule for Type I entries also can be seen as follows:  for each row of $\mathbf{Q}_{\text{v}}$ (user-delivery array in vertical 1D MACC problem), we replace each of the stars in $i^{\text{th}}$ group by  $\mathbf{H}_i$ (user-delivery array in horizontal 1D MACC problem), since $\mathbf{Q}_{\text{v}}$ has the same star entries as $\mathbf{U}_{\text{v}}$, and $\mathbf{H}_i$ has the same star entries as $\mathbf{B}_i$;  then increment the integers in $\mathbf{H}_{i}$ by the occurrence orders.
\hfill $\square$
\end{remark}

From Remark \ref{remark-2}, there are $tL(K_2-L)\times F_1'F_2$ non-star entries in Type I of $\mathbf{Q}$. From Construction~\ref{con-general-1}, each integer in Type I occurs $t$ times. So there are
$$S_{\text{I}}=\frac{tL(K_2-L)\times F_1'F_2}{t}$$
different integers filled in Type I of $\mathbf{Q}$, i.e., the server sends $S_{\text{I}}$ Type I multicast messages of packets in the first round.

\begin{example}
\label{exam-delivery-array-I}
\rm
Let us return to the example in Section~\ref{sub:illustrated example TH3} with
$(K_1,K_2,L,M,N)=(5,3,2,2,15)$, which is based on the $2$-$(6,9,6,9)$ Partition PDA $ \mathbf{H}=(\mathbf{H}_1,\mathbf{H}_2)$ in Example~\ref{ex-par}. 
 $\mathbf{H}_1$ and $\mathbf{H}_2$ are illustrated in Fig.\ref{fig-par-inte}.
\begin{table*}
\center
\caption{Fill Type I of user-delivery array $\mathbf{Q}$  with $K_1=5$, $K_2=3$, $L=2$ and $t=2$, where $U_{i,[3]}=\{U_{i,1},U_{i,2},U_{i,3}\}$ represents the set of users $U_{i,1}$, $U_{i,2}$, and $U_{i,3}$.  \label{tab-delivery-1-2-users-I}}
\begin{tabular}{|c|c|c|c|c|}
\hline
$U_{1,[3]}$& $U_{2,[3]}$& $U_{3,[3]}$&$U_{4,[3]}$& $U_{5,[3]}$\\ \hline
$\mathbf{H}_{1}$&$\mathbf{H}_{1}+9$&$\mathbf{H}_{2}$&
$\mathbf{H}_{2}+9$&\\ \hline
$\mathbf{H}_{1}+18$&$\mathbf{H}_{1}+27$&&$\mathbf{H}_{2}+18$&$\mathbf{H}_{2}+27$\\ \hline
&$\mathbf{H}_{1}+36$&$\mathbf{H}_{1}+45$&$\mathbf{H}_{2}+36$&
$\mathbf{H}_{2}+45$\\ \hline
\end{tabular}\\[0.3cm]
\end{table*}
Then the following 2D MACC user-delivery array $\mathbf{Q}$ can be obtained in Table \ref{tab-delivery-1-2-users-I}.

More precisely, 
since the number of rows in $\mathbf{U}_{\text{v}}$  is $F_1'=3$,  and in each row of $\mathbf{U}_{\text{v}}$, there are $L=2$ stars replaced by $\mathbf{B}_1$, thus
there are $6$ arrays $\mathbf{B}_1$ in $\mathbf{U}$.
 In the first   row, we fill the null entries in Type I by replacing the first $\mathbf{B}_1$ of $\mathbf{U}$ by $\mathbf{H}_1+(1-1)\times9=\mathbf{H}_1$, and replacing the second $\mathbf{B}_1$ of $\mathbf{U}$ by $\mathbf{H}_1+(2-1)\times9=\mathbf{H}_1+9$.
In the second  row, we  replace the first $\mathbf{B}_1$ of $\mathbf{U}$ by $\mathbf{H}_1+(3-1)\times9=\mathbf{H}_1+18$, and replace the second $\mathbf{B}_1$ of $\mathbf{U}$ by $\mathbf{H}_1+(4-1)\times9=\mathbf{H}_1+27$.
In the third  row, we  replace the first $\mathbf{B}_1$ of $\mathbf{U}$ by $\mathbf{H}_1+(5-1)\times9=\mathbf{H}_1+36$, and replace the second $\mathbf{B}_1$ of $\mathbf{U}$ by $\mathbf{H}_1+(6-1)\times9=\mathbf{H}_1+45$.
In the similar way,  we replace the $6$ arrays $\mathbf{B}_2$ in $\mathbf{U}$.
 As a result, we used $9\times6=54$ integers in Type I which equals $S_{\text{I}}$.
\hfill $\square$
\end{example}

\subsubsection{Step 2. Fill the null entries in Type II}
\
\newline \indent
Next, we fill the outer structure null entries in Type II.
Recall that our hybrid scheme is combined with two 1D MACC problems in vertical and horizontal projections.
In $(K_1,L,M_1,N)$ vertical 1D MACC problem, the row index is denoted by $j\in[F_1']$; and column index is denoted by $k_1\in[K_1]$.
In $(K_2,L,M_2,N)$ horizontal 1D MACC problem, the row index is denoted by $\mathbf{f}\in[K_2]^t$ (same as Partition PDA detailed in Section~\ref{sub:ori model}); and the column index is denoted by $k_2\in[K_2]$.
Thus in the $(K_1,K_2,L,M,N)$ 2D MACC problem, we define the row index of $\mathbf{U}$ and $\mathbf{Q}$ as
\begin{eqnarray}
	\label{eq-row-index-Q}
	(j,{\bf f} ),\ \ \ \text{where} \ j\in{F_1'}, \ {\bf f}\in[K_2]^t,
\end{eqnarray}
and the column index of $\mathbf{U}$ and $\mathbf{Q}$ as
\begin{eqnarray}
	\label{eq-column-index-Q}
	(k_1,k_2),\ \ \ \text{where} \ k_1\in[K_1], \ k_2\in [K_2].
\end{eqnarray}

In addition, the following notations are useful to fill the null entries in Type II.
For any integer $s\in [S'_1]$, we assume that the integer $s$ occurs $g_s$ times in $\mathbf{Q}_{\text{v}}$, i.e., $\mathbf{Q}_{\text{v}}(j_1,k_{1,1})=\mathbf{Q}_{\text{v}}(j_2,k_{1,2})=\ldots=\mathbf{Q}_{\text{v}}(j_{g_s},k_{1,g_s})=s$. 
Without loss of generality, we assume that $k_{1,1}<k_{1,2}<\cdots<k_{1,g_s}$. Notice that, all these entries containing $s$ are distributed in different rows and columns by Condition C3 of Definition~\ref{def-PDA}.
Denote the set of columns in $\mathbf{Q}_{\text{v}}$ containing $s$ by 
\begin{align}
\mathcal{S}_s=\{k_1,k_{1,2},\ldots,k_{1,g_s}\}.
\label{eq-mathcalS_s}	
\end{align}
Focus on the $g_s\times g_s$ sub-array of $\mathbf{Q}_{\text{v}}$ containing $s$, for each $\sigma\in[g_s]$, all the entries in row $j_{\sigma}$ are stars except the entry at column $k_{1,\sigma}$ since Condition C3 in Definition~\ref{def-PDA}. Then we define
\begin{align}
	\mathcal{P}_{s,j_{\sigma}}=\mathcal{S}_s\setminus\{k_{1,\sigma}\} \label{eq-mathcal-P}
\end{align}
to indicate the users served by the multicast message containing $s$ and able to retrieve the $j_{\sigma}^{\text{th}}$ packet.


From Remark \ref{remark-Q_v-Q}, for each $j\in[F_1']$ and each $k_1\in\mathcal{U}_j$, we define a one-to-one mapping
\begin{align}
\varphi_{j}(k_1)=i	
\label{eq-map-k1-i}
\end{align}
to indicate which sub-Partition array $\mathbf{H}_i$ replaces the star entry $\mathbf{Q}_{\text{v}}(j,k_1)$.
For instance, in the last row of Table~\ref{tab-delivery-1-2-users-I} and $k_1=3$, we have $\varphi_{3}(3)=1$ since this entry $\mathbf{H}_1+45$ is obtained by using $\mathbf{H}_1$.
Then for each $\sigma\in[g_s]$, we define
\begin{align}
\label{eq-mathcal-O}
\Lambda_{s,j_{\sigma}}=\{\lambda_1,\lambda_2,\ldots,\lambda_{g_s-1}\}=\{\varphi_{j_{\sigma}}(k_1)\ | \ k_1\in \mathcal{P}_{s,j_{\sigma}}\}
\end{align}and
\begin{align}
\label{eq-mathcal-noO}
\Lambda'_{s,j_{\sigma}}=\{\lambda_1',\lambda_2',\ldots,\lambda_{t-g_s+1}'\}=[t]\setminus\Lambda_{s,j_{\sigma}}, 	
\end{align}
which indicate the subscripts of beneficial sub-Partition PDAs by the multicast message containing ``$s$"  in row $j_{\sigma}$, and the other non-beneficial sub-Partition PDAs, respectively.

Now we are ready to introduce our  filling rule for the entries in Type II. From Remark~\ref{remark-2}, the entry with index $\left((j,{\bf f} ), (k_1,k_2)\right)$  of $\mathbf{U}$ is a Type II null entry if  $k_1 \in \overline{\mathcal{U}}_{j}$, where $\overline{\mathcal{U}}_{j}=[K_1]\setminus\mathcal{U}_{j}$. Thus we only need to design the filling rule for these entries, as
\begin{eqnarray}
	\label{eq-Q}
	\mathbf{Q}\big((j,\mathbf{f}), (k_1,k_2) \big)=(s, \mathbf{e}), \ \ \ \hbox{if}\ \ k_1\in \overline{\mathcal{U}}_{j},
\end{eqnarray}	
where
\begin{eqnarray}
	\label{eq-inte-s}
	s=\mathbf{P}(j,\psi_{j}(k_1)),
\end{eqnarray}
which is consistent with  $\mathbf{Q}_{\text{v}}$ in \eqref{eq-array-user-PDA}, i.e., the user-delivery array in the vertical 1D MACC problem; and
\begin{eqnarray}
	\label{eq-e}
	\mathbf{e} = (f_{\lambda_1},\ldots,f_{\lambda_{h-1}},k_2,f_{\lambda_h},\ldots,f_{\lambda_{g_s-1}},f_{\lambda'_1},\ldots,f_{\lambda'_{t-g_s+1}}) \in[K_2]^{t+1},
\end{eqnarray}
where $h$ satisfies that $\mathcal{S}_s[h]=k_1$.
Here, the subscripts $\lambda_1,\ldots,\lambda_{g_s-1}$ and $\lambda_1',\ldots,\lambda_{t-g_s+1}'$ are given in \eqref{eq-mathcal-O} and \eqref{eq-mathcal-noO}, respectively.


By the above construction, the following lemmas can be obtained whose detailed proofs are given in Appendices~\ref{appendix-typeII-c3} and \ref{appendix-typeII-SII}.
\begin{lemma}\rm
\label{lem-typeII-C3}
In the hybrid scheme, any entry defined by \eqref{eq-Q} in Type II satisfies Condition C3 of Definition \ref{def-PDA}.
\hfill $\square$
\end{lemma}

\begin{lemma}\rm
\label{lem-typeII-SII}
In the hybrid scheme, Type II entries defined by $(s,\mathbf{e})$ satisfying each integer $s\in[S_1']$ occurs at least once; and each vector $\mathbf{e}\in[K_2]^{t+1}$ occurs at least once.
\hfill $\square$
\end{lemma}

From Lemma~\ref{lem-typeII-SII}, there are
$$S_{\text{II}}=S_1'\times K_2^{t+1}=S_1'F_2K_2$$
different vectors filled in Type II of $\mathbf{Q}$, i.e., the server sends $S_{\text{II}}$ Type II multicast messages of packets in the first round.

In conclusion, the server totally sends
$S_{\text{I}}+ S_{\text{II}}$ multicast messages in the first round during the delivery phase. Furthermore, the hybrid scheme contains $K_1$ rounds. With $F=F_1' F_2\times K_1$, we can compute that the transmission load of the hybrid scheme  is
\begin{align}
\label{eq-proof-R3}	
R_{3}&= \frac{K_1\left(S_{\text{I}}+ S_{\text{II}}\right) }{F} = \frac{tL(K_2-L)}{t}+\frac{S_1'K_2}{F_1'} = \frac{K_2tL-tL^2}{t}+\frac{K_2S_1'}{F_1'},
\end{align}
which coincides with Theorem~\ref{th-General-regular}.

\begin{example}
\label{exam-delivery-array-II}
\rm
We return to the Example \ref{exam-delivery-array-I}.
From~\eqref{eq-inte-s} and Example~\ref{ex-1D-CWLZC}, for the integer $s=1$, we have
\begin{align*}
 \psi_{1}(5)=3, \ \ \ \ \  \psi_{2}(3)=2, \ \ \ \ \   \psi_{3}(1)=1.
\end{align*}
Thus $\mathcal{S}_1=\{1,3,5\}$.
From \eqref{eq-mathcal-P}, we have
\begin{align*}
	\mathcal{P}_{1,1}=\{1,3\}, \ \ \ \ \  \mathcal{P}_{1,2}=\{1,5\}, \ \ \ \ \   \mathcal{P}_{1,3}=\{3,5\}.
\end{align*}
Then combining with Example \ref{exam-delivery-array-I}, \eqref{eq-mathcal-O} and \eqref{eq-mathcal-noO}, we have
\begin{align*}
	\lambda_{1,1}=\{1,2\}, \ \ \ \ \  \lambda_{1,2}=\{1,2\}, \ \ \ \ \   \lambda_{1,3}=\{1,2\}
\end{align*}
and $\lambda'_{1,1}=\lambda'_{1,2}=\lambda'_{1,3}=\emptyset$.

When $j=1$, $\mathbf{f}=(3,2)$, $k_1=5$, and $k_2=1$, we consider the entry in the row indexed by $\big(1,(3,2)\big)$ and the column indexed by $(5,1)$. Since $k_1=5\in \overline{\mathcal{U}}_{1}=\{5\}$, from \eqref{eq-Q}, \eqref{eq-inte-s} and \eqref{eq-e}, we have
$$\mathbf{Q}\big((1,(3,2)),(5,1)\big)=\big(1,(3,2,1)\big).$$
Notice that, $k_1=5$ is the third element of set $\mathcal{S}_1=\{1,3,5\}$, thus the vector $\mathbf{e}=(3,2,1)$ is generated by appending $k_2=1$ into the third coordinate of the vector $(f_{o_1},f_{o_2})=(f_1,f_2)=(3,2)$. Similarly, we have
\begin{align*}
&\mathbf{Q}\big((2,(3,1)),(3,2)\big)=\big(1,(3,2,1)\big),\\
&\mathbf{Q}\big((3,(2,1)),(1,3)\big)=\big(1,(3,2,1)\big).
\end{align*}
The sub-array, which contains $\big(1,(3,2,1)\big)$  in Table \ref{tab-delivery-1-2-users-II},  is exactly the sub-array $\mathbf{Q}_{(3,2,1)}$ in Fig.~\ref{fig-Q-sub}\footnote{For convenience, we omitted the integer $s=1$ in Fig.~\ref{fig-Q-sub}.}, and satisfies    Condition C3 of Definition \ref{def-PDA}.
\begin{table}
\center
\caption{The sub-array containing $\big(1,(3,2,1)\big)$ \label{tab-delivery-1-2-users-II}}
\begin{tabular}{|c|c|c|c|}
  \hline
 \backslashbox{Row index}{Column index} &$(1,3)$  & $(3,2)$& $(5,1)$\\ \hline
$(1,(3,2))$                &*                  &*                  &$(1,(3,2,1))$\\ \hline
$(2,(3,1))$                  &*                  &$(1,(3,2,1))$&* \\ \hline
$(3,(2,1))$                  &$(1,(3,2,1))$&*                  &*\\ \hline
\end{tabular}
\end{table}
\hfill $\square$
\end{example}

\section{Conclusion}
\label{sec-conclusion}
 In this paper, we formulated a new 2D MACC system, which is a generalization of the existing 1D MACC system.
 A baseline 2D MACC scheme was first proposed by directly extending 1D MACC schemes to the 2D model via an MDS precoding.  When $K_1$ and $K_2$ are both divisible by $L$, we proposed an improved scheme via a grouping method. For the case where  $K_1\geq K_2>L$, we propose a new transformation approach
 to construct a hybrid 2D MACC scheme by  using two classes of 1D MACC schemes  as outer and inner structures.
On-going works include the derivation on the converse bounds for the 2D MACC model and the extension of the proposed schemes to more general 2D cellular networks,   hierarchical networks, and combination networks.



\appendices

\section{Proof of Lemma~\ref{lem-typeII-C3}}
\label{appendix-typeII-c3}
For any two different entries in Type II,  $\mathbf{Q}\left((j,{\bf f}),(k_1,k_2)\right)=(s,\mathbf{e})$ and $\mathbf{Q}\left((j',{\bf f}'),(k_1',k_2')\right)
=({s'},\mathbf{e}')$,
we have
\begin{align*}
\begin{array}{ll}
	k_1\in \overline{\mathcal{U}}_j, \ \ \ \ \ \ & k_1'\in \overline{\mathcal{U}}_{j'}; \\
	{\bf f}=(f_1,f_2,\ldots,f_t), \ \ \ \ \ \ & {\bf f}'=(f_1',f_2',\ldots,f_t'); \\
	s=\mathbf{P}(j,\psi_{j}(k_1)), \ \ \ \ \ \ \ & {s'}=\mathbf{P}(j',\psi_{j'}(k_1')); \\	
	{\bf e}=(e_1,e_2,\ldots,e_{t+1}), \ \ \ \ \  & {\bf e}'=(e_1',e_2',\ldots,e_{t+1}'),
\end{array}
\end{align*}
from \eqref{eq-Q}, \eqref{eq-inte-s} and \eqref{eq-e}. 
Furthermore, focus on the sub-arrays containing $s$ and $s'$. By simplifying \eqref{eq-mathcalS_s} and \eqref{eq-mathcal-P}, we have
\begin{align*}
	\begin{array}{ll}	
	    \mathcal{S}_s=\{s_1,s_2,\ldots,s_{g_s}\} \ \ \ \ \ \ & \mathcal{S}'_{s'}=\{s_1',s_2',\ldots,s_{g_s'}'\}; \\
		\mathcal{P}_{s,j}=\mathcal{S}_s\setminus\{k_1\},  \ \ \ \ \ \ \ & \mathcal{P}_{s',j'}=\mathcal{S}_{s'}'\setminus\{k_1'\},
	\end{array}
\end{align*}
where $\mathcal{S}_s$ and $\mathcal{S}'_{s'}$ indicate the sets of columns containing integer $s$ and $s'$, respectively; $\mathcal{P}_{s,j}$ and $\mathcal{P}_{s',j'}$ indicate the sets of columns containing stars in row $j$ and row $j'$, respectively.

Assume that
$$\mathbf{Q}\left((j,{\bf f}),(k_1,k_2)\right)=\mathbf{Q}\left((j',{\bf f}'),(k_1',k_2')\right),$$
thus we have 
\begin{align*}
s=s', \ \ \mathcal{S}_s=\mathcal{S}_{s'}',\ \ \text{and} \ \ {\bf e}={\bf e}'.
\end{align*}

It was proved in \cite{CWLZC} that user-delivery array for the outer structure $\mathbf{Q}_{\text{v}}$ satisfies Condition C3 the original PDA $\mathbf{P}$ satisfies Condition C5 in Remark~\ref{remark-C5}.
Assume that $\mathbf{P}$ satisfies Condition C5, thus $\mathbf{P}(j,k)=\mathbf{P}(j',k')=s$. Furthermore, for some integers $i_1,i_2\in[t+1]$, we have $k=\mathcal{A}_{j}\bigcup\{k\}[i_1]$, $k'=\mathcal{A}_{j'}\bigcup\{k'\}[i_2]$ where $\mathcal{A}_{j}$ and $\mathcal{A}_{j'}$ are the sets of columns containing stars in row $j$ and row $j'$ of $\mathbf{P}$; and  $k+(i_1-1)(L-1)\in\mathcal{U}_{j'}$, $k'+(i_2-1)(L-1)\in\mathcal{U}_{j}$ where $\mathcal{U}_{j'}$ and $\mathcal{U}_{j}$ are the sets of columns containing stars in row $j'$ and row $j$ of $\mathbf{Q}_{\text{v}}$.
Then let $k=\psi_j\{k_1\}$ and $k'=\psi_{j'}\{k_1'\}$. From \eqref{eq-array-user-PDA}, we have $\mathbf{Q}_{\text{v}}(j,k_1)=\mathbf{Q}_{\text{v}}(j',k'_1)=s$.
From the construction of the vertical 1D MACC scheme in Section~\ref{subsect-line}, we have $k_1=k+(i_1-1)(L-1)$ and $k_1'=k'+(i_2-1)(L-1)$ for some integers $i_1,i_2\in[t+1]$.
Thus $k_1\in\mathcal{U}_{j'}$ and $k_1'\in\mathcal{U}_{j}$ holds, i.e., $\mathbf{Q}_{\text{v}}(j,k_1')=\mathbf{Q}_{\text{v}}(j',k_1)=*$.
Hence,  $\mathbf{Q}_{\text{v}}$ satisfies Condition C3 in Definition~\ref{def-PDA}.

Now we will prove $\mathbf{Q}$ satisfies Condition C3. Since $\mathbf{Q}_{\text{v}}$ satisfies Condition C3, we have $k_1\neq k_1'$ and $\mathbf{Q}_{\text{v}}(j, k_{1}')=\mathbf{Q}_{\text{v}}(j', k_{1})=*$.
In the construction of the hybrid scheme in 2D MACC model, from Remark~\ref{remark-Q_v-Q}, these stars of $\mathbf{Q}_{\text{v}}$ are replaced based on arrays $\mathbf{H}_i$ and $\mathbf{H}_{i'}$ in $\mathbf{Q}$ respectively, where $i=\varphi_{j}(k_1')$ and $i'=\varphi_{j'}(k_1)$ are defined by \eqref{eq-map-k1-i} (i.e., from $\mathbf{Q}_{\text{v}}$ to $\mathbf{Q}$, the star $\mathbf{Q}_{\text{v}}(j, k_{1}')$ is replaced based on array $\mathbf{H}_i$, and the star $\mathbf{Q}_{\text{v}}(j', k_{1})$ is replaced based on array $\mathbf{H}_{i'}$). 	
To prove $\mathbf{Q}$ satisfies Condition C3,  it is equivalent to prove that $\mathbf{Q}\left((j,{\bf f}),(k_1',k_2')\right)$ and $\mathbf{Q}\left((j',{\bf f}'),(k_1,k_2)\right)$ are stars of $\mathbf{H}_i$ and $\mathbf{H}_{i'}$, respectively.
Without loss of generality, we assume that $k_1'<k_1$. 
From \eqref{eq-mathcal-P}, we have $\mathcal{P}_{s,j'}\bigcup\{k_1'\}=\mathcal{P}_{s,j}\bigcup\{k_1\}=\mathcal{S}_s$.
From \eqref{eq-e}, we have $\mathcal{S}_s[h']=k_1'$ and $\mathcal{S}_s[h]=k_1$ where $h'<h$, since $\mathcal{S}_s$ is sorted in an increasing order and $k_1'<k_1$.
Thus, $\mathcal{P}_{s,j}[h']=\mathcal{S}_s[h']=k_1'$; $\mathcal{P}_{s,j'}[h-1]=\mathcal{S}_s[h]=k_1$ holds.
Since $\mathbf{e}=\mathbf{e}'$, we have $e_{h'}'=e_{h'}$, which leads to $k_2'=f_{\lambda_{h'}}=f_{\varphi_{j}(k_1')}=f_i$; and $e_h=e_h'$, which leads to $k_2=f_{\lambda_{h-1}'}'=f_{\varphi_{j'}(k_1)}'=f_{i'}'$. Hence, $k_2'=f_i$ and $k_2=f_{i'}'$ always hold, such that $\mathbf{H}_i\left({\bf f},k_2'\right)=*$ and	$\mathbf{H}_{i'}\left({\bf f}',k_2\right)=*$ hold from \eqref{Eqn_Gen._Par_PDA}, where these stars are tag-stars in Definition~\ref{rem-par}. 
As a result, $\mathbf{Q}\left((j,{\bf f}),(k_1',k_2')\right)=\mathbf{Q}\left((j',{\bf f}'),(k_1,k_2)\right)=*$ holds, i.e., Condition C3 of Definition~\ref{def-PDA} is satisfied.	
For instance, the sub-array of $\mathbf{Q}$  containing $(s,\mathbf{e})=(1,(3,2,1))$ is shown in Table~\ref{tab-delivery-1-2-users-II}.

\section{Proof of Lemma~\ref{lem-typeII-SII}}
\label{appendix-typeII-SII}
From  the transformation approach detailed in Section~\ref{subsect-line} and \eqref{eq-array-user-PDA}, all the integers $s\in[S_1']$ in the original PDA $\mathbf{P}$ are filled in the user-delivery array $\mathbf{Q}_{\text{v}}$ of the vertical 1D MACC problem, and eventually filled in Type II entries by \eqref{eq-Q} and \eqref{eq-inte-s}. Since $\mathbf{P}$ satisfies Condition C2 of Definition~\ref{def-PDA}, each integer $s\in[S_1']$ occurs at least once in $(s,\mathbf{e})$.

Next, we focus on the vector $\mathbf{e}\in[K_2]^{t+1}$ in \eqref{eq-e}. 
\begin{itemize}
	\item Any  $\mathbf{e}=(e_1,e_2,\ldots,e_{t+1})\in[K_2]^{t+1}$ can be  written as $\mathbf{e}=(e_1,\ldots,e_{g_s}, e_{g_s+1}, \ldots,$ $e_{t+1})$. From \eqref{eq-e}, there exists an $h\in[g_s]$, such that $\mathcal{S}_s[h]=k_1$ where $k_1\in[K_1]$ and $e_h=k_2\in[K_2]$. 
	Thus for any given vector $\mathbf{e}$, there exists a column index $(k_1,k_2)$, where $k_1\in[K_1]$ and $k_2\in[K_2]$, which satisfies $\eqref{eq-column-index-Q}$.
	
	\item We define $\widetilde{\mathbf{f}}$ as the sub-vector of $\mathbf{e}$ by removing the coordinate ${e_h}$. From \eqref{eq-e} we have $\widetilde{\mathbf{f}}=(f_{\lambda_1},\ldots,f_{\lambda_{g_s-1}},f_{\lambda'_1},\ldots,f_{\lambda'_{t-g_s+1}})$. From \eqref{eq-mathcal-P}, \eqref{eq-mathcal-O}, \eqref{eq-mathcal-noO}, all the $t$ subscripts in $\widetilde{\mathbf{f}}$ are different from each other.
	Thus, by adjusting the order of coordinates in $\widetilde{\mathbf{f}}$, we can obtain the vector $\mathbf{f}=(f_1,\ldots,f_t)\in[K_2]^t$. In addition, from the non-star entries filled by integer $s$, we have $j\in[F_1']$. Thus for any given vector $\mathbf{e}$, there exists a row index $(j,\mathbf{f})$, where $j\in[F_1']$ and $\mathbf{f}\in[K_2]^t$, which satisfies \eqref{eq-row-index-Q}.
\end{itemize}
Hence, each vector $\mathbf{e}\in[K_2]^{t+1}$ occurs at least once  in $(s,\mathbf{e})$.

\bibliographystyle{IEEEtran}
\bibliography{grid}

\begin{thebibliography}{10}
\providecommand{\url}[1]{#1}
\csname url@samestyle\endcsname
\providecommand{\newblock}{\relax}
\providecommand{\bibinfo}[2]{#2}
\providecommand{\BIBentrySTDinterwordspacing}{\spaceskip=0pt\relax}
\providecommand{\BIBentryALTinterwordstretchfactor}{4}
\providecommand{\BIBentryALTinterwordspacing}{\spaceskip=\fontdimen2\font plus
\BIBentryALTinterwordstretchfactor\fontdimen3\font minus
  \fontdimen4\font\relax}
\providecommand{\BIBforeignlanguage}[2]{{%
\expandafter\ifx\csname l@#1\endcsname\relax
\typeout{** WARNING: IEEEtran.bst: No hyphenation pattern has been}%
\typeout{** loaded for the language `#1'. Using the pattern for}%
\typeout{** the default language instead.}%
\else
\language=\csname l@#1\endcsname
\fi
#2}}
\providecommand{\BIBdecl}{\relax}
\BIBdecl

\bibitem{GGTTG}
G.~S. Paschos, G.~Iosifidis, M.~Tao, D.~Towsley, and G.~Caire, ``The role of
  caching in future communication systems and networks,'' \emph{IEEE Journal on
  Selected Areas in Communications}, vol.~36, no.~6, pp. 1111--1125, 2018.

\bibitem{MN}
M.~A. Maddah-Ali and U.~Niesen, ``Fundamental limits of caching,'' \emph{IEEE
  Transactions on Information Theory}, vol.~60, no.~5, pp. 2856--2867, 2014.

\bibitem{yufactor2TIT2018}
Q.~Yu, M.~A. Maddah-Ali, and A.~S. Avestimehr, ``Characterizing the rate-memory
  tradeoff in cache networks within a factor of 2,'' \emph{IEEE Transactions on
  Information Theory}, vol.~65, no.~1, pp. 647--663, 2019.

\bibitem{WTP}
K.~Wan, D.~Tuninetti, and P.~Piantanida, ``An index coding approach to caching
  with uncoded cache placement,'' \emph{IEEE Transactions on Information
  Theory}, vol.~66, no.~3, pp. 1318--1332, 2020.

\bibitem{YMA}
Q.~Yu, M.~A. Maddah-Ali, and A.~S. Avestimehr, ``The exact rate-memory tradeoff
  for caching with uncoded prefetching,'' \emph{IEEE Transactions on
  Information Theory}, vol.~64, no.~2, pp. 1281--1296, 2018.

\bibitem{YCTC}
Q.~Yan, M.~Cheng, X.~Tang, and Q.~Chen, ``On the placement delivery array
  design for centralized coded caching scheme,'' \emph{IEEE Transactions on
  Information Theory}, vol.~63, no.~9, pp. 5821--5833, 2017.

\bibitem{STD}
K.~Shanmugam, A.~M. Tulino, and A.~G. Dimakis, ``Coded caching with linear
  subpacketization is possible using {R}uzsa-{S}zem{\'e}redi graphs,'' in
  \emph{2017 IEEE International Symposium on Information Theory (ISIT)}, 2017,
  pp. 1237--1241.

\bibitem{SDLT}
K.~Shanmugam, A.~G. Dimakis, J.~Llorca, and A.~M. Tulino, ``A unified
  {R}uzsa-{S}zemer{\'e}di framework for finite-length coded caching,'' in
  \emph{2017 51st Asilomar Conference on Signals, Systems, and Computers},
  2017.

\bibitem{SZG}
C.~Shangguan, Y.~Zhang, and G.~Ge, ``Centralized coded caching schemes: A
  hypergraph theoretical approach,'' \emph{IEEE Transactions on Information
  Theory}, vol.~64, no.~8, pp. 5755--5766, 2018.

\bibitem{TR}
L.~Tang and A.~Ramamoorthy, ``Coded caching schemes with reduced
  subpacketization from linear block codes,'' \emph{IEEE Transactions on
  Information Theory}, vol.~64, no.~4, pp. 3099--3120, 2018.

\bibitem{YTCC}
Q.~Yan, X.~Tang, Q.~Chen, and M.~Cheng, ``Placement delivery array design
  through strong edge coloring of bipartite graphs,'' \emph{IEEE Communications
  Letters}, vol.~22, no.~2, pp. 236--239, 2018.

\bibitem{CKSB}
H.~H.~S. Chittoor, P.~Krishnan, K.~V. Sushena~Sree, and M.~V.~N. Bhavana,
  ``Subexponential and linear subpacketization coded caching via projective
  geometry,'' \emph{IEEE Transactions on Information Theory}, pp. 1--1, 2021.

\bibitem{CJWY}
M.~Cheng, J.~Jiang, Q.~Wang, and Y.~Yao, ``A generalized grouping scheme in
  coded caching,'' \emph{IEEE Transactions on Communications}, vol.~67, no.~5,
  pp. 3422--3430, 2019.

\bibitem{CJTY}
M.~Cheng, J.~Jiang, X.~Tang, and Q.~Yan, ``Some variant of known coded caching
  schemes with good performance,'' \emph{IEEE Transactions on Communications},
  vol.~68, no.~3, pp. 1370--1377, 2020.

\bibitem{CJYT}
M.~Cheng, J.~Jiang, Q.~Yan, and X.~Tang, ``Constructions of coded caching
  schemes with flexible memory size,'' \emph{IEEE Transactions on
  Communications}, vol.~67, no.~6, pp. 4166--4176, 2019.

\bibitem{MW}
J.~Michel and Q.~Wang, ``Placement delivery arrays from combinations of strong
  edge colorings,'' \emph{IEEE Transactions on Communications}, vol.~68,
  no.~10, pp. 5953--5964, 2020.

\bibitem{ZCJ}
X.~Zhong, M.~Cheng, and J.~Jiang, ``Placement delivery array based on
  concatenating construction,'' \emph{IEEE Communications Letters}, vol.~24,
  no.~6, pp. 1216--1220, 2020.

\bibitem{CWLZC}
M.~Cheng, K.~Wan, D.~Liang, M.~Zhang, and G.~Caire, ``A novel transformation
  approach of shared-link coded caching schemes for multiaccess networks,''
  \emph{IEEE Transactions on Communications}, vol.~69, no.~11, pp. 7376--7389,
  2021.

\bibitem{SB}
S.~Sasi and B.~Sundar~Rajan, ``Multi-access coded caching scheme with linear
  sub-packetization using {PDAs},'' \emph{IEEE Transactions on Communications},
  pp. 1--1, 2021.

\bibitem{ER}
E.~Peter and B.~S. Rajan, ``Coded caching with shared caches from generalized
  placement delivery arrays,'' \emph{arXiv preprint arXiv:2107.00361}, 2021.

\bibitem{ENR}
E.~Peter, K.~K. Namboodiri, and B.~S. Rajan, ``A secretive coded caching for
  shared cache systems using {PDAs},'' \emph{arXiv preprint arXiv:2110.11110},
  2021.

\bibitem{JHNS}
J.~Hachem, N.~Karamchandani, and S.~N. Diggavi, ``Coded caching for multi-level
  popularity and access,'' \emph{IEEE Transactions on Information Theory},
  vol.~63, no.~5, pp. 3108--3141, 2017.

\bibitem{SPE}
B.~Serbetci, E.~Parrinello, and P.~Elia, ``Multi-access coded caching: gains
  beyond cache-redundancy,'' in \emph{2019 IEEE Information Theory Workshop
  (ITW)}, 2019, pp. 1--5.

\bibitem{RK}
K.~S. Reddy and N.~Karamchandani, ``Rate-memory trade-off for multi-access
  coded caching with uncoded placement,'' \emph{IEEE Transactions on
  Communications}, vol.~68, no.~6, pp. 3261--3274, 2020.

\bibitem{SR}
S.~{Sasi} and B.~{Sundar Rajan}, ``An improved multi-access coded caching with
  uncoded placement,'' \emph{arXiv e-prints}, p. arXiv:2009.05377, Sep. 2020.

\bibitem{RKstructure}
K.~S. Reddy and N.~Karamchandani, ``Structured index coding problem and
  multi-access coded caching,'' \emph{CoRR}, vol. abs/2012.04705, 2020.

\bibitem{LWCG}
D.~Liang, K.~Wan, M.~Cheng, and G.~Caire, ``Multiaccess coded caching with
  private demands,'' \emph{arXiv preprint arXiv:2105.06282}, 2021.

\bibitem{NRprivacy}
K.~K. Namboodiri and B.~S. Rajan, ``Multi-access coded caching with demand
  privacy,'' \emph{arXiv preprint arXiv:2107.00226}, 2021.

\bibitem{NRsecure}
------, ``Multi-access coded caching with secure delivery,'' \emph{arXiv
  preprint arXiv:2105.05611}, 2021.

\bibitem{OG}
E.~Ozfatura and D.~G{\"u}nd{\"u}z, ``Mobility-aware coded storage and
  delivery,'' \emph{IEEE Transactions on Communications}, vol.~68, no.~6, pp.
  3275--3285, 2020.

\bibitem{KMR}
D.~Katyal, P.~N. Muralidhar, and B.~S. Rajan, ``Multi-access coded caching
  schemes from cross resolvable designs,'' \emph{IEEE Transactions on
  Communications}, vol.~69, no.~5, pp. 2997--3010, 2021.

\bibitem{MKR}
P.~N. Muralidhar, D.~Katyal, and B.~S. Rajan, ``Improved multi-access coded
  caching schemes from cross resolvable designs,'' \emph{arXiv preprint
  arXiv:2102.01372}, 2021.

\bibitem{MKRmn}
------, ``{Maddah-Ali-Niesen} scheme for multi-access coded caching,''
  \emph{arXiv preprint arXiv:2101.08723}, 2021.

\bibitem{FP}
F.~Brunero and P.~Elia, ``Fundamental limits of combinatorial multi-access
  caching,'' \emph{arXiv preprint arXiv:2110.07426}, 2021.

\bibitem{MV1979}
V.~H. Mac~Donald, ``Advanced mobile phone service: The cellular concept,''
  \emph{The bell system technical Journal}, vol.~58, no.~1, pp. 15--41, 1979.

\end{thebibliography}

\end{document}